\documentclass[runningheads]{llncs}
\usepackage[T1]{fontenc}
\usepackage{graphicx}
\usepackage{orcidlink}

\usepackage{amsmath,amssymb,wasysym}
\usepackage{proof}

\hypersetup{colorlinks,linkcolor={blue},citecolor={black},urlcolor={blue}}

\usepackage[shortlabels,inline]{enumitem}

\usepackage{xurl}
\usepackage{cite}

\setlist[itemize]{topsep=0pt, itemsep=0pt}

\usepackage[colorinlistoftodos,prependcaption,textsize=small]{todonotes}

\newcommand{\kevin}[1]{\todo[linecolor=blue,backgroundcolor=blue!25,bordercolor=blue]{KK: #1}}

\input{macros.tex}

\usepackage[capitalise]{cleveref}
\newcommand\myproof[1]
   {}

\newcommand\alb\allowbreak

\newcommand\compl\sqsubseteq
\newcommand\uncompl\sqsupseteq
\newcommand\complS\sqsubset  

\mathchardef\mhyphen="2D
\newcommand{\leftOut}[1]{}

\newbox\boxA

\newcommand{\oexp}{\mbox{\hphantom{$+_{\mathsf{o}}$}}\llap{$\text{\textasciicircum}_{\mathsf{o}}\kern.1em$}}

\definecolor{light-gray}{gray}{0.85}
\newcommand\highlight[1]{\mbox{\setlength{\fboxsep}{0pt}\colorbox{light-gray}{\strut$#1$}}}
\newcommand\hlt\highlight
\newcommand\hhighlight[1]{\mbox{\scriptsize \setlength{\fboxsep}{0pt}\colorbox{light-gray}{\strut$#1$}}}
\newcommand\hhlt\hhighlight

\newcommand{\cexp}{\mbox{\hphantom{$+_{\mathsf{o}}$}}\llap{$\text{\textasciicircum}_{\mathsf{c}}\kern.1em$}}

\newcommand\TC[1]{\mathsf{#1}}
\newcommand\CC[1]{\mathsf{#1}}

\newcommand\CHOPFROMUN{.25}
\newcommand\UN{{\setbox\boxA=\hbox{\_}\usebox\boxA\kern-\CHOPFROMUN\wd\boxA{\color{white}\vrule height 0ex depth .444ex width \CHOPFROMUN\wd\boxA}\kern-\CHOPFROMUN\wd\boxA}}

\newtheorem{mylemma}{Lemma}

\newtheorem{prop}[mylemma]{Prop}
\newtheorem{thmm}[mylemma]{Thm}
\newtheorem{defi}[mylemma]{Def}

\newtheorem{coro}[mylemma]{Corollary}

\newcommand\te\simeq

\newcommand{\sm}{\smallsetminus}

\renewcommand{\phi}{\varphi}

\renewcommand{\iff}{\allowbreak\mathrel{\;\leftarrow\nobreak\kern-1.6ex\rightarrow\;}\allowbreak} 

\newcommand{\la}{\leftarrow}

\newcommand{\ra}{\rightarrow}
\newcommand{\Ra}{\Rightarrow}
\newcommand{\lra}{\longrightarrow}

\def\implies{\lra} 

\newcommand{\llam}{{\llam}}

\newcommand\tab{\hspace*{3ex}}

\newcommand{\smobla}{{{\CC{smobla}}}}

\newcommand{\coverageTest}{{\CC{coverage\hspace*{-0.1ex}Test}}}

\newcommand{\decrease}{{\CC{decrease}}}

\newcommand{\Compat}{{{\CC{Compat}}}}

\newcommand{\Const}{{{\CC{C
}}}}

\newcommand{\Term}{{{\CC{Term}}}}
\newcommand{\FTerm}{{{\CC{FTerm}}}}

\newcommand{\CTerm}{{{\CC{CTerm}}}}

\newcommand{\arOf}{{{\CC{arOf}}}}

\newcommand{\ctpOf}{{{\CC{ctpOf}}}}

\newcommand{\mgen}{{{\CC{mgen}}}}

\newcommand{\mtpOf}{{{\CC{mtpOf}}}}

\newcommand{\tpOf}{{{\CC{tpOf}}}}

\newcommand{\choice}\varepsilon 

\newcommand{\Var}{\CC{{Var}}}

\newcommand{\Type}{\CC{{Type}}}

\newcommand{\TVar}{\CC{{TVar}}}

\newcommand{\KK}{{\TC{K}}}

\newcommand{\pickPos}{{\mathit{pickPos}}}

\newcommand{\Nat}{\mathbb{N}}

\newcommand{\bool}{\TC{bool}}
\newcommand{\nat}{\TC{nat}}

\newcommand{\listt}{\TC{list}}

\newmuskip\originalthinmuskip
\originalthinmuskip=\thinmuskip
\newmuskip\tinyGapMu
\renewcommand\ldots{\mathinner{.\mskip\originalthinmuskip .\mskip\originalthinmuskip .}}

\newcommand\Smash[1]{\kern-200mm\smash{#1}\kern-200mm}
\newcommand\SubItem[2]{\indent\hbox to \leftmargini{\hfill#1\enskip}#2}

\newcommand\XDot{\raise1ex\hbox{\Large.\kern.1em}}

\def\figurename{Fig\hbox{.}}

\newcommand{\er}{{{\mathsf{erase}}}}

\newcommand{\FTV}{{{\mathsf{FTV}}}}
\newcommand{\TV}{{{\mathsf{TV}}}}

\newcommand{\Poss}{{{\mathsf{Poss}}}}

\newcommand{\scmp}{\,\bullet\,}

\newif\ifwithappendix
\withappendixtrue

\crefname{section}{\S}{\S\S}
\Crefname{section}{\S}{\S\S}
\crefformat{section}{\S#2#1#3}
\Crefformat{section}{\S#2#1#3}
\makeatletter
\AddToHook{cmd/appendix/before}{%
  \crefalias{section}{appendix}%
  \crefalias{subsection}{appendix}}
\makeatother
\crefname{appendix}{App.}{App.}
\Crefname{appendix}{App.}{App.}

\begin{document}
\title{Just Type It in Isabelle! AI Agents Drafting, Mechanizing,~and~Generalizing~from~Human~Hints}
\titlerunning{Just Type It in Isabelle!}

\author{
Kevin Kappelmann\inst{1}\orcidlink{0000-0003-1421-6497} \and
Maximilian Schäffeler\inst{2}\orcidlink{0000-0002-2612-2335} \and
Lukas Stevens\inst{3}\orcidlink{0000-0003-0222-6858} \and
Mohammad~Abdulaziz\inst{2}\orcidlink{0000-0002-8244-518X} \and
Andrei Popescu\inst{1}\orcidlink{0000-0001-8747-0619}\and
Dmitriy Traytel\inst{3}\orcidlink{0000-0001-7982-2768}}
\authorrunning{Kappelmann et al.}

\institute{
Department of Computer Science, University of Sheffield, United Kingdom
\email{\{k.kappelmann,a.popescu\}@sheffield.ac.uk}\and
Department of Informatics, King's College London, United Kingdom
\email{\{maximilian.schaffeler,mohammad.abdulaziz\}@kcl.ac.uk}\and
Department of Computer Science, University of Copenhagen, Denmark
\email{\{lukas.stevens,traytel\}@di.ku.dk}}
\maketitle              
%
\begin{abstract}
Type annotations are essential when printing 
terms in a way that preserves their meaning under reparsing and type inference. We study the problem of complete and minimal type annotations for rank-one polymorphic $\lambda$-calculus terms, as used in Isabelle. Building on prior work by Smolka, Blanchette et al., we give a metatheoretical account of the problem, with a full formal specification and proofs, and formalize 
it in \isabellehol. Our development is a series of experiments 
featuring human-driven and AI-driven formalization workflows: a human and an LLM-powered AI agent independently produce pen-and-paper proofs, and the AI agent autoformalizes both in Isabelle, with further human-hinted 
AI interventions refining and generalizing the development.

\keywords{Printing-parsing roundtrip \and Minimal type annotation \and AI \and Large language model \and Autoformalization \and Isabelle \and Lambda calculus}
\end{abstract}
\section{Introduction}\label{sec:intro}
Some 30 years ago, Larry Paulson taught the working ML\footnote{Hereby, we refer to the programming language Standard ML, and its spiritual successors, implementations, and variants such as OCaml, Poly/ML and Isabelle/ML.} programmers how to pretty print~\cite{DBLP:books/daglib/0084777}. His textbook algorithm is mainly concerned with the printing layout for improved legibility, and has inspired the pretty printers of types, terms, and theorems used in the Isabelle proof assistant to this day.

When printing Isabelle's terms, or more generally terms of the typed lambda calculus (\cref{sec-synPrelim}), in addition to legibility, it is crucial to preserve the type information. Specifically, we are interested in the following round-trip property: starting from a fully typed term, we wish to print it in a way so that reparsing it and using type inference afterwards yields the original term. For example, printing the term $(c_\alpha,d_\alpha)$\footnote{We prefer the compact $t_\tau$ over Isabelle's $t :: \tau$ type annotation syntax.} with polymorphic constants $c$ and $d$ as $(c,d)$ violates this property: invoking type inference yields the more general typing $(c_\alpha,d_\beta)$.

Achieving the round-trip hence requires type annotations.
As legibility remains a desideratum, annotations should be as sparse as possible while still constraining type inference to recover the original term. Smolka, Blanchette et al.~\cite{\smoblacitations}\pagebreak[2] identified these requirements and formulated the type annotation problem in the context of Isar proof reconstruction for Sledgehammer~\cite{PaulsonB10}. They call the round-trip property \emph{completeness}, and
they are interested in a solution that is also \emph{minimal}.
We use completeness and minimality as the \emph{correct printing problem}'s
specification~(\cref{sec-annotate}).
Smolka, Blanchette et al.~\cite{\smoblacitations} also devised a solution to the correct printing problem,
indicating that full minimality (w.r.t.\ a function assigning costs to annotation positions)
is NP-hard, they settle for a
tractable greedy algorithm, producing a complete, \emph{locally minimal} solution.
Their implementation is part of the Isabelle distribution and has been reused in different contexts in which terms are printed: \texttt{sketch} and \texttt{explore}~\cite{sketch_and_explore}, \texttt{super\_sketch}~\cite{DBLP:conf/icse-formalise/TanDMW25}, AutoCorres2~\cite{DBLP:journals/afp/BrecknellGHIKKLNSSSTW24}, types-to-sets~\cite{DBLP:conf/cpp/Milehins22}, and Apply2Isar~\cite{binder2026apply2isarautomaticallyconvertingisabellehol}.

Despite widespread usage, the algorithm's implementation contained a com\-pleteness-compromising issue in one optimization. In the Isabelle2025-2 release, $(c_\alpha,d_\alpha)$ is printed as $(c,d)$---the type annotation is incorrectly dropped. In addition, it is difficult to convince Isabelle to actually print type annotations correctly when they are attached to constants and bound variables: $[]_{\nat\;\listt}$ gets printed as $[]$ and $\lambda x.\,x_\nat$ as the garbled \texttt{\_type\_constraint\_}. We resolved these issues by correcting the optimization, forbidding annotations on bound variables in favor of annotating binders,
and adding a custom print translation that preserves all annotations inserted by the algorithm. We also set up a testing environment for the algorithm using Mirabelle~\cite{DBLP:conf/itp/DesharnaisVBW22}. These implementation problems, however, also exposed a deeper issue: without a precise formal account of what the algorithm is supposed to guarantee, it is difficult to explain exactly why an optimization is unsound, to justify the repairs, or to rule out similar regressions.

Alas, this paper is not about our technical refinements of the implementation. Our main goal as working programming-language metatheorists is to understand the type annotation problem precisely, with formal statements and complete proofs. (Smolka, Blanchette, et al.~\cite{\smoblacitations} give only informal claims and no proofs.) At the same time, we also identify as working metatheory formalizers and thus want a proof assistant to validate our proofs. Being modern working metatheorists and formalizers, we would like to conduct both activities assisted by modern technology, particularly large language models (LLMs), thus corroborating numerous recent accounts of their usefulness in our field.

To this end, we conduct the following case studies contrasting human-driven and AI-driven activities, starting from Smolka, Blanchette, et al.'s informal description, the Isabelle/ML implementation, and an example algorithm execution:
\begin{enumerate}
\item A human expert formalizes the type annotation problem's metatheory on paper in \cref{sec-synPrelim,sec-annotate,sec:human_paper_proof}, an occasion to describe the problem in detail.
\item Independently, an LLM-powered AI agent formalizes the 
problem's metatheory on paper. We discuss the experimental setup and how the result differs from the human pen-and-paper formalization (\cref{sec:llm_paper_proof}). Drafts were iteratively refined through human reviews. Some feedback also went in the other direction: AI's work has informed the human formalizer to prove local minimality.
\item A second AI agent (disconnected from the first) formalizes\pagebreak[2] both pen-and-paper proofs 
 in \isabellehol. Both resulting autoformalizations
 use mostly the same setup; we discuss differences in the experience (\cref{sec:autoformalization}).
\item A human expert intervenes with the AI's pen-and-paper formalization,
observing that a core part of the algorithm could be generalized by reducing it to a standard problem, for which a rich theory as well as off-the-shelf algorithms exist. The AI agent analyses the 
human hint, shows the reduction to be correct, and updates
the Isabelle formalization to do the same (\cref{sec:independence_system}).
\end{enumerate}

To our own surprise, \emph{all} case studies were successful after a small number of human reviews.
We obtained three Isabelle formalizations of the desired results without writing a single line of Isabelle code ourselves.
As our LLM, we used \claude 4.6, a state-of-the-art general purpose model with strong reasoning capabilities.
Claude proved a capable, although not logically infallible, companion---one that can substantially amplify the work of human experts.

In contrast to prior work in autoformalization (\cref{sec-related}),
which typically starts from pen-and-paper proofs and/or targets classic mathematics,
we begin with an algorithm that lacks formal correctness specifications.
Another 
contribution is autoformalization of programming language metatheory,
which has been largely unexplored.
Finally, we are among the first to report on a comprehensive autoformalization in \isabelle and provide a simple, reusable setup to the community.
Our experimental data~\cite{kappelmann_2026_19406624} and an extended version~\cite{kappelmann2026justtypeisabelleai} are available online.



\section{Syntactic Preliminaries}
\label{sec-synPrelim}

Next we introduce the necessary background: the syntax of types and terms, and the
 typing relation for rank-one polymorphic $\lambda$-calculus. Our presentation is tailored to the forthcoming discussion of the Smolka-Blanchette type annotation removal algorithm, in that our notion of terms is more general than usual, allowing annotations at any position in the term's abstract syntax tree.
The material in \cref{sec-synPrelim,sec-annotate,sec:human_paper_proof}
was \emph{not} used for the AI pen-and-paper proofs (\cref{sec:llm_paper_proof}).

\subsection{Types}
\label{subsec-types}

For the entire paper,
we fix an 
infinite set $\TVar$ of {\em type variables} (\emph{tyvars} for short), ranged over by $\alpha,\beta$, and an
infinite set $\Var$ of {\em term variables} (\emph{vars} for short), ranged over by $x,y,z$.
We also fix a
{\em type structure}, \ie a pair $(\KK,\arOf)$ where
	$\KK$, ranged over by $\kappa$,  is the set of {\em type constructors}, and
	$\arOf : \KK \ra \Nat$, is a function associating
	\emph{arities} to the type constructors.

The {\em types}, 
ranged over
by $\sigma,\tau$, forming the set $\Type$, are given
by the
grammar
$ 
\sigma \;::= \; \alpha \,\mid\, \sigma \Ra \tau \,\mid\,  (\sigma_1,\ldots,\sigma_{\arOf(\kappa)})\,\kappa
$. 
Thus, a type is 
a tyvar, or the function type constructor $\Ra$ applied to two types,
or another 
type constructor $\kappa$ postfix-applied to a number of types 
matching its arity.


Finally, we fix a
{\em signature} for the type structure $(\KK,\arOf)$, \ie a pair $\Sigma = (\Const,\ctpOf)$, where:
$\Const$, ranged over by $c$, is a set of symbols called {\em constants}, and
$\ctpOf : \Const \ra \Type$ is a function associating a type to every constant.

A \emph{(type) substitution} is a function $\rho : \TVar \ra \Type$ with finite support, in that $\rho(\alpha) \not= \alpha$ for all but a finite set of tyvars. For a substitution $\rho$ and a type $\sigma$,
$\sigma[\rho]$ denotes the application of $\rho$ to $\sigma$.
For a type $\sigma$, $\TV(\sigma)$ denotes the set of its (occurring) tyvars.
We say that $\sigma$ is an \emph{instance} of $\tau$ via a substitution $\rho$, written $\sigma \leq_\rho \tau$, when
	$\sigma = \tau[\rho]$; and that $\sigma$ is an \emph{instance} of $\tau$, or that
	$\tau$ is \emph{more general} than $\sigma$,
	written $\sigma \leq \tau$, when there exists $\rho$ such that $\sigma \leq_\rho \tau$.
	%
For two substitutions $\rho$ and $\rho'$, their \emph{composition} $\rho \scmp \rho'$ is defined by $(\rho \scmp \rho')(\alpha) = \rho'(\alpha)[\rho]$.
We write $\sigma/\alpha$ for the substitution that takes $\alpha$ to $\sigma$ and any other tyvar to itself,
%
and $\rho[\alpha \la \sigma]$ for the
substitution obtained by updating $\rho$ at $\alpha$ with $\sigma$.

To model (partial) type-annotations, we will work with elements of the set $\Type_\bot = \Type \cup \{\bot\}$ (where $\bot\notin \Type$), which we call \emph{maybe-types}, and let $\xi,\zeta$ range over them. The operators $\TV$ and $\_[\_]$ are extended from $\Type$ to $\Type_\bot$ 
by defining $\TV(\bot) = \emptyset$ and $\bot[\rho] = \bot$.

\subsection{Terms}
\label{subsec-terms}

The
\emph{partially typed terms}  (which we simply call \emph{terms}),
ranged over by $s,t$, forming the set $\Term$,
are given by the grammar
$t \;::=\; x_\xi \,\mid\,  c_\xi \,\mid\, (t_1\,t_2)_\xi \,\mid\,  (\lambda x_\xi.\,t)_\zeta$.
We call $x_\xi$ a \emph{(maybe-)typed variable} and $c_\xi$ a \emph{constant instance}.
Thus, a term is either
\begin{enumerate*}[(1)]
\item a typed variable,
\item a constant instance,
\item an application decorated with a maybe-type, or
\item a $\lambda$-abstraction of a typed variable, decorated with a maybe-type.
\end{enumerate*}
%
Our terms are ``decorated'' with maybe-types on any occurring variable or constant, as well as at any subterm.
We think of
an actual type decoration (when the maybe-type is a type) as a \emph{type annotation}, and of
a $\bot$ decoration as the absence of a type annotation.
Therefore, in examples we will sometimes omit some $\bot$ decorations, 
thus writing, \eg
$\lambda x_\alpha.\;y_\beta$ instead of $(\lambda x_\alpha.\;y_\beta)_\bot$.

Our 
term decoration scheme
captures the process of type inference---of which we think of as ``completing'' a term such as $\lambda x_{\alpha\Ra\beta \Ra \gamma}.\;x\;c$
to a fully annotated term such as $(\lambda x_{\alpha\Ra\beta \Ra \gamma}.\,(x_{\hhlt{\alpha\Ra\beta \Ra \gamma}}\,c_{\hhlt{\alpha}})_{\hhlt{\beta\Ra\gamma}})_{\,\hhlt{(\alpha\Ra\beta \Ra \gamma) \Ra \beta\Ra\gamma}}$, where the highlighted types have been inferred.
%
Additionally, this scheme allows for a smooth presentation of the algorithm that we will study,
which, starting from a fully annotated term, repeatedly removes redundant annotations.

We let $\TV(t)$ be the set of tyvars occurring in $t$, \ie occurring in all the type annotations from $t$;
and $t[\rho]$ be the application of the substitution $\rho$ to $t$, which is defined by applying $\rho$ to (all the types occurring in) $t$.  We let $\FTV(t)$ be the set of free typed variables occurring in $t$; e.g., 
$\FTV(\lambda x_\sigma.\;y_{\sigma \Ra \tau}\,x_\sigma) = \{y_{\sigma \Ra \tau}\}$.

We extend the ``instance of'' relation from types to terms: $t \leq_\rho s$ is defined as $t = s[\rho]$, and $t \leq s$ is defined as the existence of $\rho$ such that $t \leq_\rho s$; in this case, we say $t$ is an \emph{instance} of $s$, or that $s$ is \emph{more general} than $t$.
\leftOut{
We say that $t$ and $s$ are \emph{equally general}, written $s \equiv t$, when $s \leq t$ and $t \leq s$.
(Thus, $\equiv$ is the equivalence standardly generated by the preorder $\leq$.)
We also write $s < t$ for the notion of being \emph{strictly more general}, defined by $s < t$ iff $s \leq t$ and $s \not\equiv t$. (Note that this is not the same as $s \leq t$ and $s \not= t$; this later version would not properly capture the intended concept.)
} 

A term will be called:  
\begin{itemize}
	\item \emph{unambiguous}, if its binding variable occurrences are non-repetitive,
	in that there exists no variable $x$ and (possibly equal) maybe-types $\xi$ and $\zeta$ such that
	$\lambda x_\xi$ and $\lambda x_\zeta$ occur at two different positions in a term;
	\item a \emph{fully typed term} (\emph{F-term}), when all its occurring maybe-types are types;
	\item a \emph{Church-typed term} (\emph{C-term}), when
	(1) all maybe-types annotating its application and abstraction subterms are $\bot$, and
	(2) all maybe-types annotating its constants and variables, as well as its binding variables, are types.
\end{itemize}
The above concepts have straightforward inductive definitions. 
(We focus on unambiguous terms in order to ensure correct behavior of type substitution.) 

We let $\FTerm$, 
and $\CTerm$
denote the subsets of $\Term$ consisting of the F-terms 
and C-terms,
respectively.
We let $u,v$ range over F-terms.

A \emph{position} is a list of numbers in $\{1,2\}$,
identifying the location of a subterm or a binding variable in a term via its unique path.
$\Poss(t)$ denotes the set of positions of term $t$.
If $p \in \Poss(t)$, 
we write $\mtpOf(t,p)$ for the maybe-type decoration at that position.
\ifwithappendix
(Details in \cref{app-detailsSyntax}.)
\fi
For example, if $t$ is $(\lambda x_\nat.\;(f_\bot\,x_\bot)_{\bool})_\bot$, then
	$\Poss(t) = \{[],[1],[2],[2,1],[2,2]\}$;
	$\mtpOf(t,[]) = \bot$,
	$\mtpOf(t,[1]) = \nat$, $\mtpOf(t,[2]) = \bool$,
	and $\mtpOf(t,[2,1]) = \mtpOf(t,[2,2]) = \bot$.
%
We write $\mtpOf(t)$ instead $\mtpOf(t,[])$. 
For F-terms $u$ we write $\tpOf(u,p)$  and $\tpOf(u)$
instead of $\mtpOf(u,p)$ and $\mtpOf(u)$ (since we know that this is a type).

On $\Type_\bot$, we define the relation $\compl$ by $\xi \compl \zeta$ iff
$\xi \in \{\bot,\zeta\}$.  This order is extended to $\Term$, by defining the \emph{annotation subsumption} relation
$t \compl s$ to mean that $s$ is obtained from $t$ by adding zero or more type annotations: 
\begin{gather*}
	\infer{x_\xi \compl x_\zeta}
	{\xi \compl \zeta}
	\qquad
	\infer
	{c_\xi \compl c_\zeta}
	{\xi \compl \zeta}
	\qquad
	\infer{(s\;t)_\xi \compl (s'\;t')_{\xi'}}
	{s \compl s' & t \compl t' &  \xi \compl \xi' }
	\qquad
	\infer
	{(\lambda x_\zeta.\;t)_\xi \compl (\lambda x_{\zeta'}.\;t')_{\xi'}}
	{\xi \compl \xi' & t \compl t' & \zeta \compl \zeta'}
\end{gather*}
%
\leftOut{
We write $\complS$ for the strict version of $\compl$, defined by $s \complS t$ iff
$s \compl t$ and $s \not= t$.
}

For a term $t$ and position $p\in \Poss(t)$, let $t[p:= \bot]$ denote the term obtained from
$t$ by erasing
the type annotation at $p$ (\ie turning its decoration into $\bot$).
For a term $t$, the 
(completely unannotated) term $\er(t)$ denotes the term obtained from $t$ by erasing \emph{all} its type annotations. 
Note that $\er(t) \compl t[p:= \bot] \compl t$.

The \emph{well-typedness} predicate $\vdash$ is defined on $\FTerm$ by the following rules:
%
\begin{gather*}
\infer{\vdash x_\sigma}{}
\qquad \qquad \qquad \qquad \qquad
\infer{\vdash c_\sigma}{\sigma \leq \ctpOf(c)}
\\
\infer{\vdash (u\,v)_\sigma}{\vdash u & \vdash v & \tpOf(u) = \tpOf(v) \Ra \sigma}
\qquad
\infer{\vdash (\lambda x_\sigma.\,u)_{\sigma\Ra \tpOf(u)}}{\vdash u & \forall x_\tau \in \FTV(u).\; \tau = \sigma}
\end{gather*}

Note that the condition $\forall x_\tau \in \FTV(u).\; \tau = \sigma$
guarantees that terms with type-incompatible bindings are ill-typed,
such as $\lambda x_\alpha.\;x_\beta$. Alternatively, this would also be achieved
by the explicit consideration of typing contexts.

An F-term $u$ is said to be a \emph{well-typed completion} of a term $t$,
provided  $t \compl u$ and $\vdash u$.
We call a term $t$ \emph{typable} when it has a well-typed completion.
Note that for F-terms, typability is equivalent to well-typedness.

\leftOut{
\begin{remark}
\begin{enumerate}[(1)]
\item Type inference for a C-term amounts to finding a well-typed completion of it.\kevin{that's true for all terms though}
\item Curry-style\kevin{I don't know what that is.} type inference for a U-term amounts to finding a most general well-typed completion of it.
\item We can model Isabelle terms as unambiguous C-terms. (Indeed, we can use fresh variables for abstractions, which makes the term unambiguous.)
\item In general, we do not want to allow ambiguous terms such as $\lambda x_\nat.\;\lambda x_\bool.\;x_\bot$ because one cannot infer a most general type for them. In all proofs, we will implicitly use the fact that unambiguity is hereditary:  if a term is unambiguous, then so are its subterms.
\item As a matter of ``formal hygiene'', even though we define the substitution and type-variable operators,  $\_[\_]$ and $\TV$, on arbitrary terms (and prove basic properties of them also on arbitrary terms), for the critical results we only use these, as well as the generality relation $\leq$, on F-terms only.
\end{enumerate}
\end{remark}
}


In our formalism, the process of type inference for a (partially annotated) term $t$ means producing a most general well-typed completion for $t$.
We will take for granted the classic Damas-Hindley-Milner type inference result \cite{hindley1969principal, milner1978theory, damas1985type,damas-milner:1982}, which we only need in existential form:
%
\begin{thmm}\rm   \label{prop-assumed} 	
	Assume that $t$ is a typable unambiguous term.
	Then there exists a most general
	well-typed completion of $t$ (i.e., a $\leq$-maximal
	F-term among the F-terms $u$ such that
	$t \compl u$ and $\vdash u$).
\end{thmm}

We let $\mgen(t)$ denote a choice (unique up to tyvar renaming) of a most general well-typed completion of $t$.

\section{The Problem of Correct Printing and the Smolka-Blanchette Algorithm}
\label{sec-annotate}

The terms used internally by Isabelle are captured by our notion of unambiguous C-terms.
A crucial property of these terms is that they contain enough information in order for their inferred type to be actually unique:

\begin{prop}\rm   \label{prop-Cterm-uniqueCompleteness}
	Assume that $t$ is a typable unambiguous C-term.
	Then there is only one well-typed completion of $t$ (hence only one choice for $\mgen(t)$).
\end{prop}

We can now define a correct printing of an Isabelle term to be another term (not necessarily a C-term)  that has the same type-free content and allows for the same unique most general typing as the original:

\begin{defi}\rm \label{defi-correctPrint}
	Given a typable unambiguous C-term $t$,
	a \emph{correct printing} of $t$ is a term $s$ such that $\mgen(t)$
	is the unique most general well-typed completion of $s$.
\end{defi}

We are interested in a correct printing that is minimal w.r.t.\ the annotation subsumption relation, $\compl$. This is also the problem that Smolka and Blanchette have considered, leading to their  implementation which is currently part of Isabelle.  They have not formulated the definition of correct printing rigorously, but focused on more practical, operational matters. Namely, they
described a criterion for deeming an annotation redundant,
provided an implementation that starts with a fully annotated term and repeatedly removes positions deemed redundant for as long as possible following a reverse greedy pattern, and
informally claimed minimality. Next, we discuss a simplified version of their algorithm and a detailed informal proof that it indeed achieves a minimal correct printing. 

The difference between what we will describe and what Smolka and Blanchette
actually
implemented
is threefold: (1) we abstract away from the specific cost function that they employ when deciding the annotation at which of the eligible positions to remove, (2) we ignore a number of small optimizations,   
and (3) we ignore typing contexts altogether. Concerning (2), in future work we plan to extend our formal development (and our partnership with LLMs) to cover these optimizations as well, and possibly to even go beyond them together with a further sharpening of the implementation.
Concerning (3), our motivation for avoiding contexts is that they are completely uninteresting
for the correct printing problem---in that the tyvars (from $\TV(\sigma_i)$) and variables ($x_i$) from
a typing context
$\Gamma = x_1:\sigma_1,\ldots,x_n:\sigma_n$
would be treated exactly as if they were part of the signature $\Sigma$, namely as nullary type constructors and 
constants; so the problem can be safely reduced to printing closed terms in the empty context.
%
\leftOut{
One important caveat we should mention here: While we do not pose this restriction when stating our results in this section, the final theorems (on correctness and minimality of printing), while true for arbitrary terms, are only meaningful when restricted to closed terms.  The correct extension of these results to open terms should only be considered along the aforementioned route, namely extending the signature to incorporate $\Gamma$.
}


We express the Smolka-Blanchette algorithm (subject to the simplifications discussed above) by the function
$\smobla$ which we define below, after defining the auxiliary functions $\coverageTest$ and $\decrease$.

The ternary predicate $\coverageTest$ acting on F-terms $v \in \FTerm$, terms $s \in \Term$ and positions $p \in \{1,2\}^*$,
is defined by
\begin{center} \begin{tabular}{l}
		$\coverageTest(v,s,p) \;=$
		\\\tab$p \in \Poss(s) \cap \Poss(v) \wedge \mtpOf(s,p) \not= \bot \;\wedge $
		\\\tab$(\forall \alpha \in \TV(\tpOf(v,p)).\;$
		\\\hspace*{10ex}$\exists q \in \Poss(s) \sm \{p\}.\;\alpha \in \TV(\tpOf(v,q)) \wedge \mtpOf(s,q) \not= \bot)$
\end{tabular} \end{center}
As will be seen from the upcoming definitions,
$\coverageTest$ will be called with $v$ being $\mgen(\er(t))$, the most general well-typed completion of the erasure of the original term $t$, and $s$ being the current term obtained by repeated removal of position annotations from $\mgen(t)$; so we will always have
$s \compl v$, which (as we discuss in \cref{subsec-proofDevel}) will ensure $ \Poss(s) = \Poss(v)$ and $\mtpOf(s,q) \compl \tpOf(v,q)$. The predicate checks whether a position $p$ on the one hand has an annotation within $s$, and on the other hand has all the tyvars occurring in $v$ at that position
 covered by another annotation in $s$, in the sense that there is another position $q$ that has an annotation within $s$ and $\alpha$ occurs in $v$ at $q$.

We say that a function
$\pickPos : \FTerm \times \Term \ra \{1,2\}^*$ is $\coverageTest$-compatible
when, for all $v\in \FTerm$ and $s\in\Term$, if $\exists p.\;\coverageTest(v,s,p)$ then
$\coverageTest(v,s,\pickPos(v,s))$. Thus, compatibility means being a correct choice function
for the position argument of $\coverageTest$.
We let $\Compat$ denote the set of $\coverageTest$-compatible functions.

Now, the function $\decrease : \Compat \times \FTerm \times \Term \ra \Term $ is defined by
\begin{center} \begin{tabular}{l}
		$\decrease(\pickPos,v,s) \;= $
		$\left\{\begin{array}{ll}
			\decrease(\pickPos,v,s[(\pickPos(v,s)):= \bot]), &
			\\ \mbox{\hspace*{5ex}if $\exists p.\;\coverageTest(v,s,p)$} &
			\\
			s, \mbox{ otherwise}
		\end{array} \right.$
\end{tabular} \end{center}


Thus, $\decrease$ keeps removing annotations from $s$ at positions chosen using $\pickPos$ provided they pass the coverage test.
Finally, the function $\smobla : \Compat \times \Term \ra \Term $ is defined by calling $\decrease$ on $v = \mgen(\er(t))$ (the fixed provider of tyvars that must stay covered) and $s = \mgen(t)$ (the changing term from which annotations keep being removed, initially set to the fully annotated term $\mgen(t)$).
\begin{center} \begin{tabular}{l}
		$\smobla(\pickPos,t) = \decrease(\pickPos,\,\mgen(\er(t)),\,\mgen(t))$
\end{tabular} \end{center}

In the above algorithm, we can distinguish two components.
First, there is the general-purpose reverse greedy component: The algorithm
	starts with a fully annotated term $\mgen(t)$ and, via its $\decrease$ function, keeps removing annotations for positions $p$ satisfying a given test, $\coverageTest$, for as long as such annotation carrying positions exist (and the choice of the exact position to remove is regulated by a $\coverageTest$-compatible parameter function $\pickPos$).

Then, there is  the ad hoc component, given by the definition of $\coverageTest$, which is the heart of the algorithm. Intuitively, we are interested in removing annotations at positions for as long as we do not lose information about the (unique) typing of the term.  The
algorithm does not pursue
this (clear yet non-effective) intuition directly, but in a roundabout manner: It first erases all type annotations from $t$ and builds a most general well-typed completion for the erased term, $v = \mgen(\er(t))$. Then it proceeds based on the following crucial observation: The most general well-typed completion of the original term, $s = \mgen(t)$, is less than or equally general as $v$,
which means that the tyvars of $v$ correspond position-wise
to specific subterms of $s$.
%
%
Protection against loss of typing information is offered through the test $\coverageTest$,
which makes sure that any $\alpha\in \TV(v)$ is still ``covered'' by an annotation in $s$ at a corresponding position, in that there is a position $p$ such that (1) $\alpha$ appears in the type annotation at that position in $v$, and (2) $p$ has a type annotation in $s$ as well.
Then the fact that the algorithm achieves
(i) correct printing that is also (ii) minimal (which is the subject of our proof development)
essentially amounts to the fact that this test
(i) guarantees the desired preservation of typing information and (ii) nothing beyond that.

For brevity, our presentation of the algorithm did not formally separate the generic from the ad hoc  component. However, for the sake of conceptual cleanness and reusability, a mechanization should ideally do that. In \cref{sec:independence_system} we report on how an AI agent has performed this separation based on a minimal prompt.

\section{Human-Authored Paper Proof Development}\label{sec:human_paper_proof} 

In this section we give a rigorous pen-and-paper description of
our results about the type annotation problem,
which we also call \emph{minimal correct printing problem},
starting with the statements
(\cref{subsec-statementsMainRes}) and then delving into the proofs (\cref{subsec-proofDevel}).
We aim to provide enough detail for the reader to form an informed view of the statement and proof complexity involved in what we consider a fairly direct ``no-nonsense'' solution, and thereby to appreciate the extent of the creative effort an AI agent would require when starting from first principles.

\subsection{The statements of the main results}
\label{subsec-statementsMainRes}

Our first goal will be to prove that the Smolka-Blanchette algorithm returns
a correct printing of the original term (something that Smolka and Blanchette refer to as ``completeness''):

\begin{thmm}[Completeness]\rm  \label{thm-correctPrintingHuman}
	If $t$ is a typable unambiguous C-term and $\pickPos \in \Compat$,
	then $\smobla(\pickPos,t)$ is a correct printing for $t$.
\end{thmm}
Then we will prove that what is being returned is a minimal such correct printing:

\begin{thmm}[Minimality]\rm  \label{thm-correctPrintingMinimalHuman}
	If $t$ is a typable unambiguous C-term and $\pickPos \alb \in \Compat$,
	then $\smobla(\pickPos,t)$ is $\compl$-minimal among the correct printings for $t$.
\end{thmm}

One may object to the notion of correct printing, hence to the statement of completeness, as being too weak.
Indeed, $s$ being a correct printing of $t$ means that $\mgen(t)$ is the unique
\emph{most general} well-typed completion of $s$---which makes sense because this is what
the type inference algorithm produces. However, one may argue, we want $\mgen(t)$ to be the unique
well-typed completion of $s$ (without the ``most general'' qualifier), because we do not want any alternative way to type the term sneaking in; let us call this (superficially) stronger notion a \emph{strong correct printing}.  We 
show that the two notions coincide, because the existence of two well-typed completions implies the existence of two most general ones.


\begin{prop}\rm \label{lem-liftMinimalityToMostGen}
	Assume $s\in \Term$ is unambiguous, and $u,u' \in \FTerm$ are two distinct well-typed completions of $s$. Then there exist two distinct \emph{most general} well-typed completions of $s$.
\end{prop}


This immediately gives:

\begin{coro}\rm  \label{coro-correctPrintingMinimalHuman}
Assume $s\in \Term$ is unambiguous. Then $s$ is a correct printing of $t$ if and only if $s$ is a strongly correct printing of $t$.
\end{coro}

\subsection{Proof development}
\label{subsec-proofDevel}

Next we give an overview of our proof development leading to the above results,
namely Thm.~\ref{thm-correctPrintingHuman}, Thm.~\ref{thm-correctPrintingMinimalHuman} and
Prop.~\ref{lem-liftMinimalityToMostGen}, as well as Prop.~\ref{prop-Cterm-uniqueCompleteness} which underlies the definition of correct printing.
\ifwithappendix
\Cref{app-detailsHumanProofs} gives full details.
\fi
The development is supported by several types of background lemmas:

\textbf{(1) Constructor-aware inversion lemmas} associated to inductively defined predicates such as $\vdash$ and $\compl$.\pagebreak[3]
These lemmas reconstitute ``one inductive rule back'' of history in situations
where an inductive predicate holds with one of the arguments having a specific syntactic constructor at the top. For example, the $\lambda$-abstraction aware left inversion lemma for $\compl$ states the following:
	If $(\lambda x_\zeta.\;s)_\xi \compl t$ then there exist $s'$, $\zeta'$ and $\xi'$
	such that $s \compl s'$,
	$\zeta \compl \zeta'$, $\xi\compl \xi'$ and $t= (\lambda {x_{\zeta'}}.\;s')_{\xi'}$.
The (immediate) proofs of these lemmas combine the freeness (injectiveness) of the syntactic constructors with the standard inversion (elimination) rules stemming from inductive definitions. (Incidentally, Isabelle automates the process of inferring these lemmas via the \texttt{inductive\_cases} command.)

\textbf{(2) Lemmas on syntax basics,} describing properties of the occurring-variable and substitution application operators, $\TV$ and $\_[\_]$. These lemmas are mostly the typical ones proved when developing a theory of syntax, such as substitution compositionality, $t[\rho \scmp \rho'] = t[\rho'][\rho]$,
distributivity of $\TV$ along substitution, $\TV(t[\rho]) = \bigcup_{\alpha \in \TV(t)} \!\TV(\rho(\alpha))$,
and substitution extensionality (in the Isabelle ecosystem also known as substitution congruence),
$\forall \alpha \in \TV(t).\,\rho(\alpha) = \rho'(\alpha)$ implies $t[\rho] = t[\rho']$.
We will in fact need the converse of extensionality,
which is slightly outside the usual repertoire
in theories of syntax: $t[\rho] = t[\rho']$ implies $\forall \alpha \in \TV(t).\,\rho(\alpha) = \rho'(\alpha)$.
There are both type and term variants of these lemmas, and their proofs go by straightforward
structural induction on types or terms (with the proofs of the latter using the former).

\textbf{(3) Annotation lemmas,} by which we collectively mean simple lemmas about positions, type annotations at positions, and the annotation-removal operator and their interaction with substitution and occurring tyvars---stating, for example, that positions
are not affected by substitutions, $\Poss(t[\rho]) = \Poss(t)$, that
substitution operates position-wise, 
$\mtpOf(t[\rho],p) = \mtpOf(t,p)[\rho]$,
and that deleting an annotation yields a decrease in the annotation ordering, $t[p:= \bot] \compl t$.
The proofs again proceed by straightforward
structural inductions on terms.

\textbf{(4) The type preservation lemma,} stating that substitution preserves typing, i.e., $\vdash v$ implies $\vdash v[\rho]$, proved by induction on the definition of $\vdash$.

With these preparations, the proof justifying
Prop.~\ref{prop-Cterm-uniqueCompleteness} is a low-hanging fruit.

\begin{proof}[Prop.~\ref{prop-Cterm-uniqueCompleteness}.]
Let $u$ and $v$ be such that $t \compl u$, $t \compl v$, $\vdash u$, $\vdash v$.
Then $u=v$ follows by induction on $t$, using the constructor-aware left inversion rules for $\compl$. The fact that $t$ is a C-term is used in the variable, constant, and abstraction cases, ensuring that respective annotation is 
a type (not just a maybe-type). \qed
\end{proof}
%

Moving towards the main results, we state lemmas about annotation subsumption and its interaction with the other operators. 
In particular, we need that $\compl$ is a preorder on terms,
does not affect the position sets, and interacts as expected with annotations at positions (in that $s \compl t$ implies $\mtpOf(s,p) \compl \mtpOf(s,p)$)---we will refer to these as \emph{the $\compl$-lemmas}.
Moreover, we need that 
less annotated terms yield more general completions, in that $t \compl s$ implies $\mgen(s) \leq \mgen(t)$. This fact, henceforth called 
\emph{the $(\compl,\leq)$-lemma},
supports the main intuition behind the Smolka-Blanchette algorithm---which, as discussed in
\cref{sec-annotate}, decides on 
removing 
annotations taking advantage of the
position synchronization offered by $\mgen(t) \leq \mgen(\er(t))$, and this in turn follows from $\er(t) \compl t$.


As a final preparation, we need some consequences of the (quasi-)generic
reverse greedy component of the algorithm,
namely that the result has fewer annotations than the original
and that the coverage test no longer holds for the result.

Now, the technical core of the correct printing theorem (Thm.~\ref{thm-correctPrintingHuman})
is expressed by the following ``sandwich'' property: If an F-term $u$
on the one hand subsumes the annotations of a term $s$, and on the other hand is less general than a term $v$, i.e., is $(\compl,\leq)$-sandwiched between $s$ and $v$, such that all tyvars of $v$ are covered by an annotation of $s$, then $u$ is uniquely determined by $s$ and $v$.

\begin{mylemma} \rm  \label{lem-crucial}
	Assume $s\in \Term$ is unambiguous and $u,v\in \FTerm$ such that
	$s  \compl u \leq v$ and $\forall \alpha \in \TV(v).\,\exists p \in \Poss(v).\,\alpha \in \TV(\tpOf(v,p)) \wedge \mtpOf(s,p) \not= \bot$. Then $u$ is the unique $u' \in \FTerm$ such that
	$s \compl u' \leq v$.
\end{mylemma}
\begin{proof}
	From $u \leq v$, we obtain a substitution $\rho$ such that
	$u \leq_\rho v$.
	Let $u' \in \FTerm$ be such that $s \compl u' \leq v$, i.e., (1) $s \compl u' \leq_{\rho'} v$ for some $\rho'$. We need to prove $u' = u$.
	By syntax basics, it suffices to fix
	$\alpha \in \TV(v)$ and to show that $\rho(\alpha) = \rho'(\alpha)$.

	From $\alpha \in \TV(v)$ and our assumption, we obtain $p \in \Poss(v)$ such that  (2) $\alpha \in \TV(\tpOf(v,p))$ and $\mtpOf(s,p) \not= \bot$.  With  the $\compl$-lemmas and $s\compl u,u'$, this implies (3) $\tpOf(u,p) = \mtpOf(s,p) = \tpOf(u',p)$.
	Moreover, by the annotation lemmas, we have (4) $\tpOf(u,p) = \tpOf(v[\rho],p)  = \tpOf(v,p)[\rho]$, and similarly (5) $\tpOf(u',p) = \tpOf(v[\rho'],p)  = \tpOf(v,p)[\rho']$.
	From (3), (4) and (5), we obtain
	$ \tpOf(v,p)[\rho] = \tpOf(v,p)[\rho']$. Applying syntax basics (specifically the converse of substitution extensionality) to
	this and (2) yields $\rho(\alpha) = \rho'(\alpha)$, as desired.  \qed
\end{proof}

\begin{proof}[Thm.~\ref{thm-correctPrintingHuman}]
Let $v = \mgen(\er(t))$. By $\smobla$'s definition, we have $s = \decrease(\pickPos,v,\mgen(t))$.
 We first prove $\phi(s') \implies \phi(\decrease(\pickPos,v,s'))$ for terms $s'$
 by induction on $|\Poss_{\not=\bot}(s')|$, where
{%
\setlength{\abovedisplayskip}{1pt}
\setlength{\belowdisplayskip}{1pt}
\begin{align*}
  \phi(s') = {}& \er(t) \compl s' \compl \mgen(t) \land {} \\
              & \forall \alpha \in \TV(v).\,\exists p \in \Poss(v).\;\alpha \in \TV(\tpOf(v,p)) \wedge \mtpOf(s',p) \not= \bot,\\
    \Poss_{\not=\bot}(s') = {}& \{p \in \Poss(s') \mid\mtpOf(s',p) \not = \bot\}.
\end{align*}}

   	In the inductive proof, we distinguish two cases:
   	(i) if $\exists p.\;\coverageTest(v,s',p)$, then we apply the induction hypothesis together with $\compl$-lemmas, annotation lemmas, and some straightforward set-theoretic computation;
   	(ii) otherwise, we have $\decrease(\pickPos,v,s') = s'$, and the fact holds trivially.

   Now, since $\phi(\mgen(t))$ holds, we obtain $\phi(s)$.
	By the $(\compl,\leq)$-lemma, we have $\mgen(t) \leq v$. From this and $\phi(s)$,
	by Lemma~\ref{lem-crucial} we obtain that $\mgen(t)$ is the unique F-term $u$ such that
	$s \compl u \leq v$.

	Since $\vdash \mgen(t)$ and $\forall u\in\FTerm.\;s \compl u \;\wedge\, \vdash u \implies u \leq v$ (by the $(\compl,\leq)$-lemma and the $\er(t) \compl s$ fact), we obtain that $\mgen(t)$ is the unique well-typed completion of $s$, \ie the unique F-term $u$ such that
	$\vdash u$ and $s \compl u$.
	\qed
\end{proof}

What we actually proved in Thm.~\ref{thm-correctPrintingHuman} is the a priori stronger property that $s$ is a strongly correct printing of $t$---as we did not assume $u$ to be most general.

It remains to prove Prop.~\ref{lem-liftMinimalityToMostGen} and
the minimality theorem Thm.~\ref{thm-correctPrintingMinimalHuman}.
These require further preparations.
First, a trivial lemma about observable difference
in annotations when descending on the annotation subsumption relation:

\begin{mylemma}\rm \label{lem-otherTriv}
	If $s,s'$ are terms such that $s' \complS s$ ,
	then there exists $p \in \Poss(s) = \Poss(s')$ such that $\mtpOf(s',p) = \bot \not= \mtpOf(s,p)$.
\end{mylemma}

Then a crucial lemma allowing us to change
an F-term's instance (by changing the instantiating substitution) without affecting
annotation subsumption:

\begin{mylemma}\rm \label{lem-multipleCompl}
	Assume $s$ is a term, $v$ is an F-term and $\rho,\rho'$ are substitutions
	such that 
		$s \compl v[\rho]$.
		and $\forall \alpha,p.\;p\in \Poss(v) \wedge \alpha \in \TV(v,p) \wedge \mtpOf(s,p) \not= \bot \implies \rho(\alpha) = \rho'(\alpha)$.
	Then $s \compl v[\rho']$.\end{mylemma}

Finally, we need a lemma
that allows us to construct a second well-typed most general completion
from a given one with ``loose'' tyvars:

\begin{mylemma}\rm \label{lem-twoMG}
	Assume $s\in \Term$ is unambiguous, and $v$ is a most general well-typed completion of $s$ such as $\TV(v) \sm \TV(s) \not= \emptyset$.
	Then there exists a most general well-typed completion of $s$ different from $v$.
\end{mylemma}

\begin{proof}[Prop.~\ref{lem-liftMinimalityToMostGen}]
Expanding the definitions, we have $u \not= u'$, $s\compl u$, $\vdash u$,
$s\compl u'$ and $\vdash u'$.
Let $v = \mgen(s)$.
From the definition of $\mgen$, we have $\vdash v$ and obtain the substitutions $\rho,\rho'$ such that $u = v[\rho]$ and $u' = v[\rho']$.
Since $u\not=u'$, by syntax basics we obtain $\alpha \in \TV(v)$ such that (a) $\rho(\alpha) \not= \rho'(\alpha)$.

We claim that $\alpha \notin \TV(s)$. 
Assuming $\alpha \in \TV(s)$, the annotation lemmas yield $p \in \Poss(v) = \Poss(u) = \Poss(u') = \Poss(s)$ such that (b)
$\mtpOf(s,p) \not=\bot$ and (c) $\alpha \in \TV(\mtpOf(s,p))$. From these and $s \compl v$, we have (d) $\alpha \in \TV(\tpOf(v,p))$.
From (b) and $s\compl u,u'$, we have
$\tpOf(u,p) = \tpOf(u',p) = \mtpOf(s,p)$. From this, by annotation lemmas we have $\tpOf(v,p)[\rho] = \tpOf(v,p)[\rho']$, which together with (a) and (d) yields a contradiction by syntax basics.

Hence, $\alpha \in \TV(v) \sm \TV(s)$. From this and $v$'s definition, by Lemma~\ref{lem-twoMG}, we obtain $v'$ such that $v' \not= v$ and $v'$ is a most general well-typed completion of $s$. As desired, $v$ and $v'$ are distinct most general well-typed completions of $s$.
\qed
\end{proof}

\begin{proof}[Thm~\ref{thm-correctPrintingMinimalHuman}]
	Let $s'$ be an unambiguous term such that (a) $s' \complS s$.
	Let $v = \mgen(\er(t))$.
	By  Lemma~\ref{lem-otherTriv}, there exists $p \in \Poss(s) = \Poss(s')$ such that (b) $\mtpOf(s',p) = \bot \not= \mtpOf(s,p)$.  Since by 
	generic reverse greedy lemmas we have $\neg\;\coverageTest(v,s,p)$, we obtain $\alpha \in \TV(\tpOf(v,p))$ such that $\forall q \in \Poss(s) \sm \{p\}.\;\alpha \in \TV(\tpOf(v,q)) \implies \mtpOf(s,q) = \bot$. With (a) and (b), %
	this gives us (c) $\forall q \in \Poss(s').\;\alpha \in \TV(\tpOf(v,q)) \implies \mtpOf(s',q) = \bot$.

	Let $\rho$ be such that $\mgen(t) = v[\rho]$. From this, since $\mgen(s) = \mgen(t)$ by Thm.~\ref{thm-correctPrintingHuman}, we obtain $\mgen(s) = v[\rho]$, and therefore (1) $s' \compl v[\rho]$.

	We choose a fresh type $\tau$ such that
	$\tau \not= \rho(\alpha)$ (for example, a fresh tyvar),
	and let $\rho' = \rho[\alpha \ra \tau]$. Using (c) and the definition of $\rho'$,
	we obtain
	\begin{itemize}
		\item[(2)] $\forall q\in \Poss_{\not=\bot}(s'), \beta \in \TV(v,q).\; \rho(\beta) = \rho'(\beta)$.
	\end{itemize}
	For (1) and (2), applying Lemma~\ref{lem-multipleCompl}, we obtain (3) $s' \compl v[\rho']$.

	By the type substitution lemma, from $\vdash v$ we obtain (4) $\vdash v[\rho']$.
	Since $\alpha \in \TV(v)$ and $\rho(\alpha) \not= \rho'(\alpha)$, by  syntax basics (again the converse of the substitution extensionality lemma), we obtain (5) $v[\rho] \not = v[\rho']$. %

	Thus, (1--5) tell us that $v[\rho]$ and $v[\rho']$ are two distinct
	well-typed completions of $s$. Finally, we apply Prop.~\ref{lem-liftMinimalityToMostGen} to obtain
	two distinct \emph{most general} well-typed completions of $s'$, as desired.
\qed
\end{proof}

Note that we used Prop.~\ref{lem-liftMinimalityToMostGen} (which forms the basis of the equivalence between correct printing and strong correct printing) in order to finalize our proof of minimality. This is natural, because in the proof it was easier to prove minimality in the extent of the (superficially) narrower concept of strong correct printing---which could then be lifted to minimality w.r.t.\ correct printing via Prop.~\ref{lem-liftMinimalityToMostGen}.

\paragraph{Summary of the human 
development. } Unsurprisingly for a metatheory of syntax, a sizable number of simple lemmas were needed, having straightforward inductive proofs. 
In addition, some creativity was required to pin down the intuition of why the coverage test correctly models the typing information preservation goal (via results such as the sandwich \Cref{lem-crucial} and the instance changing \Cref{lem-multipleCompl}).
Apart from the proof development per se, modeling the concepts required 
nontrivial design choices, such as working with terms decorated with maybe-types and introducing annotation subsumption $\compl$ inductively. 


\section{AI-Authored Paper Proof Development}
\label{sec:llm_paper_proof}

In parallel to 
the human-authored development described in the previous section,
we independently employed an AI agent to produce a pen-and-paper (LaTeX)
specification and proofs for the 
Smolka-Blanchette algorithm. 
We first describe
our technical setup (\cref{sec:llm_setup})
and the paper generation process (\cref{sec:llm_proof}),
and then provide an analysis and a comparison
with the human proof development (\cref{subsec-discussPendAndPaperAI}).

\subsection{Technical setup}
\label{sec:llm_setup}

Our setup is based on \emph{\opencode} \cite{opencode-v1.2.24},
an open source AI coding environment.
It provides AI agents with several tools,
such as the ability to read and edit files,
search the web,
and usage of the command line
and the Model Context Protocol (MCP).
As our agent, we used \claude~4.6-v1 \cite{anthropic-claude-opus-4.6},
a state-of-the-art model optimized for coding and reasoning tasks.

In each experiment, we provided the agent some background material
(\eg links to related work)
and an experiment-specific \texttt{AGENTS.md} file as input.
The \texttt{AGENTS.md} files contain the detailed instructions for the agent.
The files were created by the agent itself
from short hand-crafted prompts containing
the problem statement and some guidelines (take notes, make backups, etc.).
The experiments mimicked an academic peer-review process:
The agent was instructed to create a complete artifact (a LaTeX document),
which was then reviewed by human experts.
The feedback was passed to the agent, and the process was repeated.


\subsection{Generating the paper proofs}
\label{sec:llm_proof}

Generating proofs for the Smolka-Blanchette algorithm poses
challenges. 
First, their original presentation of the concepts \cite{\smoblacitations} is extremely informal.
Not only does it lack proofs, but the desirable properties are not rigorously defined.
Concerning the notion of correct printing, they introduce it by saying:
``types are inferred correctly when the generated HOL formulas
are parsed again by Isabelle''~\cite{smobla_journal}.
As for minimality, they note that the \emph{greedy algorithm}'s
``goal is to compute a locally minimal set of sites that completely covers all
type variables [in the substitution's domain]''~\cite{smobla_journal}.
Thus, the notion of completeness (complete coverage) with respect to which they formulate minimality appears not to be the ``declarative'' notion of not losing typing information (which in \cref{sec-annotate} we captured by the definition of correct printing) but simply the ``operational'' notion employed by the algorithm---namely, that of passing the coverage test ($\coverageTest$ in our \cref{sec-annotate} notation).  Completeness and minimality w.r.t.\ this operational notion follow directly from the nature of the reverse greedy algorithm.
However, the algorithm is undoubtedly aimed at achieving the declarative versions (of the kind we formulated in \Cref{sec-annotate}), but that is not made explicit in their text.
%
Another smaller challenge comes from their implementation diverging from their paper presentation 
by including several additional optimizations. 

Despite these challenges,
the agent produced a specification, with some limitations on generality,
and mostly correct proof sketches in its first attempt, as further elaborated below.
Following the initial version,
we ran three revision rounds
in which the agent resolved the limitations while also adding proof details,
resulting in a correct proof document amenable to autoformalization
(cf.\ \cref{sec:llm_autoformalization}).
We next describe the agent's input
followed by a summary of the generation process.
All artifacts are available in the supplementary material.

\paragraph{Input.}
We created a \texttt{notes.md} file, containing
\begin{enumerate*}[(1)]
\item links to previous works~\cite{\smoblacitations},
\item the corrected Isabelle/ML implementation, and
\item the execution trace of the algorithm on an example.
\end{enumerate*}
In the \texttt{AGENTS.md} file,
we instructed the agent to create results publishable at a peer-reviewed venue
and to take notes in files while iteratively working on its tasks.
The notes serve as reference for the agent,
\eg to recover its progress upon memory loss due to context compaction.

In rounds 2--4, we moreover passed as input a review by a human expert
and extra instructions to use an internal two-stage review process before finishing
the paper:
\begin{enumerate*}[(1)]
\item \emph{Self-review}: the agent creates a review for its paper and revises it.
\item \emph{Simulated peer-review}:
the agent performs simulated peer reviews for a scientific submission
and finalizes the paper based on the feedback.
\end{enumerate*}
We found that the agent identified
several presentational and some logical issues using internal reviews,
such as the wrong usage of an assumption,
the need for extra lemmas,
and missing definitions.
However, some significant shortcomings were only identified by human review, as explicated below.

\paragraph{Round 1.}
In the initial round, we prompted the agent to
create a precise, formal exposition for the unoptimized algorithm including proofs.
The agent successfully produced a paper
with a specification and proof sketches.
Here are some notable highlights, impressing a non-empty subset of the authors:
(1) While the statement of completeness is informal---``reparsing [the output] recovers exactly the types of $t$ [the input]''---
the agent showed the correct formal statement (cf. Thm.~\ref{thm-correctPrintingHuman}
and Def.~\ref{defi-correctPrint}) in its proof.
(2) The agent introduced an 
almost correct statement of minimality.
(3) The agent discovered correct proof approaches for completeness and minimality
and introduced several required lemmas,
such as converse substitution extensionality (cf. \cref{subsec-proofDevel}).
%

However, the human reviewer also identified several defects, most notably:
(1) The completeness and minimality proofs were restricted to ground terms (for no good reason),
and minimality also needed
a small correction to exclude an existential quantifier instantiation that
made the statement vacuously true.
(2) There were some self-contradictory statements, \eg ``The output of the annotation algorithm is a constraint-free term augmented with type constraints''. 
(3) Definitions for key concepts
(\eg annotated terms, well-typedness, position enumeration),
and statements of key lemmas (\eg preservation of well-typedness under substitution) were missing.

Some milder issues included the AI's affinity for the \isabelleml implementation,
resulting in a very technical presentation.
For example, it used
Isabelle 
jargon and notations, such as $\mathsf{dummyT}$, inference in ``pattern'' mode,
and a carbon copy of Isabelle's term and type syntax
(including a false claim: ``for any type $\tau$ of sort $s$, the type $\tau\; \mathsf{list}$ also has sort $s$'').
Moreover, it delved into implementation-specific performance optimizations that are irrelevant for the algorithm.

\paragraph{Rounds 2--4.}
Each revision round was initiated by a short prompt:
read the new input (notably, the human review), devise a plan, and revise.
The human review summarized the shortcomings of the previous round.
Admittedly, writing the reviews was not an overly pleasant task,
but it seemed like a necessary one:
while the agent was able to convince itself and its role-played scientific peers
that the paper was worthy of publication fairly quickly, the human controller thought otherwise.
The agent was able to correct, adapt and generalize its statements and proofs 
and to significantly improve its exposition in each revision round,
but it required substantial human work to spot and review logical flaws,
defects in generality, and deficiencies in rigor.
After four rounds, we were convinced that the document is ready for autoformalization
(though not for ``publication'').

\subsection{Discussion}
\label{subsec-discussPendAndPaperAI}

The most remarkable aspect of this experiment's outcome was the agent's ability
to make its way towards 
the correct concepts and results,
and to produce proof sketches that were plausible
(and later validated by autoformalization). Specifically, 
the agent
discovered the notion of strongly correct printing---although it did not introduce it
as a separate notion, the agent's statements of completeness and minimality are equivalent
to variations of our Thms.~\ref{thm-correctPrintingHuman} and \ref{thm-correctPrintingMinimalHuman} that have ``strongly correct printing'' instead
of ``correct printing''. Our 
 Corollary~\ref{coro-correctPrintingMinimalHuman} shows that the two concepts are equivalent, though
 this is not something that the agent has attempted (nor was it prompted) to establish.
Another remarkable aspect is the informative nature of the agent's output.
The agent's human controller learned the problems's correct specification and
the proof ideas by reading, 
and also by critically scrutinizing the generated document.
Moreover, before being informed by the agent's work, the human expert producing our pen-and-paper development from \Cref{sec-annotate}
had not realized the necessity to even state minimality (essentially due to
confusing the operational and the declarative versions, hence wrongly assuming minimality to be trivial).

%
A persistent weakness of the agent was its lack of precision,
including missing or slightly wrong definitions and lemmas, wrong assumptions, and a few logical errors.
We employed simulated peer reviews by the agent to mitigate these issues,
but many slipped through the review process,
only being identified by the human expert.
It is reported that LLMs perform worse detecting their own mistakes than those by other models~\cite{selfreview_bad},
which we will consider for future work.
All in all, the prevalence of such flaws urges the need for verification,
which we consider in \cref{sec:autoformalization}.

The agent's document, while not publication-ready,
served as an effective self-study material
that human experts can scrutinize and improve.
\leftOut{
	\begin{enumerate*}[(1)]
		\item There were flaws in the agent's proofs%
		---some blatant (\eg restriction to ground terms) and some subtle (\eg the insufficient formulation of minimality)---%
		that the human had to spot,
		requiring expertise of the underlying theory.
		\item The agent stayed closely to the existing work and \isabelleml implementation,
		even when instructed to abstract from particularities.
		This complicates the reuse of existing and invention of new concepts
		(such as maybe-types and completion relations $\compl$ from \cref{sec:human_paper_proof})
		that capture the problem's essence more elegantly.
	\end{enumerate*}
}
The human document undoubtedly trumps the agent's in terms of clarity and
conceptual cleanness.  
One qualitative difference
stems from the agent staying close to the existing work and \isabelleml implementation,
even when instructed to abstract from particularities; this complicates the reuse 
and invention of 
concepts
(such as maybe-types and completion relations $\compl$ from \cref{sec:human_paper_proof})
that capture the problem's essence more cleanly.
On the other hand, the AI development clearly 
wins cost-wise and time-wise:
The agent created an autoformalizable document
with a cost of 
\cost{70}, two hours of computation time
and 
one day of human review effort, while the human-authored document took approximately five days of human work.



\newcommand{\slop}[1]{\textcolor{blue}{#1}}
\section{Autoformalization in \isabellehol}
\label{sec:autoformalization}



Building on the results of the human and AI paper proofs, we conduct \emph{autoformalization} experiments for both approaches in \isabellehol.
The generated Isabelle theories are available on Zenodo~\cite{kappelmann_2026_19406624}\ifwithappendix, and Isabelle snippets with the final theorem statements can be found in \cref{app-detailsAutoformalization}\fi.
Our experiments are motivated by the following questions:
\begin{enumerate*}[(1)]
    \item Can the correctness of algorithms in the domain of programming language theory be formalized automatically, with \emph{reasonable costs and resources}?
    \item What are the \emph{limits} of state-of-the-art LLMs as \isabelle proof engineers?
    \item How does the autoformalization effort differ between mechanizing human expert proofs versus AI-generated proofs?
    Relatedly, can autoformalization serve as an effective tool for \emph{validating AI-generated results}?
\end{enumerate*}

\subsection{Technical setup}
\label{sec:llm_setup_isabelle}

Our setup extends our AI paper proof setup from \cref{sec:llm_setup}.
We connect \opencode to \isabelle via \emph{\iq}~\cite{isabelleq},
an MCP server for \jedit.
\iq lets the agent use \isabelle interactively,
with tool access, such as
modifying theory files,
fact search,
\sledgehammer,
and retrieval of proof states.
Inspired by prior autoformalization efforts~\cite{isatop},
we provide an \isabelle guidance file to the agent.
This guide contains both style guidelines,
\eg to create idiomatic \isar proofs,
and proof engineering 
patterns to facilitate efficient interaction with \iq.

\subsection{AI-authored proof autoformalization}\label{sec:llm_autoformalization}

First, we tasked \claude to mechanize the AI-authored proof.
%
The agent receives all input and generated artifacts from the AI paper experiment (\cref{sec:llm_paper_proof}).
Our \texttt{AGENTS.md} file instructs the agent to produce a formalization plan and implement it in \isabellehol.
We provide no problem-specific strategies for formalization.

From the initial prompt,
we obtain a \texttt{sorry}-free proof in \isabellehol that fills in details missing in the AI paper proof,
while also improving its structure:
\begin{enumerate*}[(1)]
    \item It uses an \isabelle \emph{locale}~\cite{locales} to introduce well-formedness assumptions on the inputs and to axiomatize the existence of most general types.
    \item Another locale is used to abstract the proof of the main theorems, fixing a set of locally minimal annotations covering all inference variables.
    \item To prove that each kept position is required for coverage, the agent identified a necessary, complex statement generalization, depicted in \cref{fig:gen_statement}.
\end{enumerate*}

\begin{figure}[t]
\centering
\includegraphics[alt={Isabelle code for lemma rg_fold_witness_aux. Assumes dist: distinct map fst of cands. And cnt_eq: for all alpha, cnt alpha equals the length of the filtered list of candidates where alpha is in K, plus extra alpha. And mem: the pair p, K is in the set of candidates. And kept: p is in fst of rg_fold cands cnt. Shows: there exists an alpha in K such that extra alpha equals zero, and for all p prime and K prime, if the pair p prime, K prime is in the set of candidates, and p prime is in fst of rg_fold cands cnt, and p prime does not equal p, then alpha is not in K prime.}, width=\textwidth]{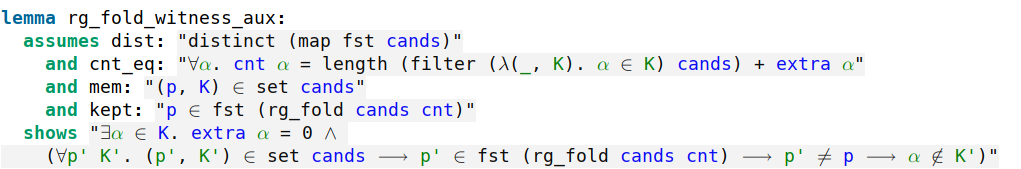}
\caption{AI-generalized statement for induction proof:
The non-generalized statement expresses that each position kept by the reverse-greedy algorithm (\texttt{rg\_fold})
must cover at least one type variable.
Here, \texttt{cands} is a list of candidate annotations with their tyvars,
\texttt{cnt} counts the occurrences of each tyvar in \texttt{cands}.
The agent noticed that an induction proof requires generalizing the statement to also count for each tyvar in how many kept positions it occurs (via \texttt{extra}).
}
\label[figure]{fig:gen_statement}
\end{figure}

However, the first human review identified several instances where the theorem statements in \isabellehol and the AI paper are not aligned:
\begin{enumerate*}[(1)]
    \item Completeness and minimality are proved only within the abstraction locale, but are not instantiated to obtain statements about the algorithm's actual output.
    \item The formalization claims to abstract from the order in which annotations are processed, but actually fixes a concrete post-order enumeration.
    This discrepancy still remains in the final version, despite human reviews repeatedly
    pointing to the problem.
    \item The AI paper uses the informal notion of \emph{consistency of a term with a set of annotations} in the algorithm's specifications.
    The formalization's definition seems to not match the paper's intended meaning,
    though the authors also argued whether the paper version is ambiguous and thus the true culprit.
\end{enumerate*}


Based on the review, the agent correctly instantiates the abstraction locale.
A second round of human reviews identified that the completeness statement
refers to the set of annotations determined by the algorithm
and not the output of the algorithm (as the paper version does).
In a final run, Claude aligns the completeness statement with the paper.
The agent also modifies the definition of consistency,
which remains suboptimal but captures the right meaning in the statements of completeness and minimality.



\subsection{Human-authored proof autoformalization}
\label{sec:human_proof_autoformalization}

As a second experiment,
we instructed the agent to mechanize the human proof,
which was explicitly created with autoformalization in mind.
%
The agent received as input a draft of \cref{sec:human_paper_proof} and the same prompt as in the first autoformalization experiment.
In round one, Claude produces \texttt{sorry}-free proofs which faithfully capture the underlying material,
except for a small but severe issue discovered in the review process:
As part of a locale definition, the agent introduces assumptions on $\coverageTest$ before defining the constant.
This allows the instantiation of these assumptions with arbitrary terms,
that could have been exploited to derive a contradiction,
but was not.
Missing concrete instantiations of the minimality and completeness statements constitute further minor issues.

In round two, the agent makes minor adjustments that address all issues raised in the review.
After the second run,
we happened to discover an improvement in the paper definition of well-typedness from \cref{subsec-terms}, namely the condition on the free typed variables from the typing of abstraction.
When informed about the change, the AI agent correctly updates the definition of well-typedness in the formalization and adapts the proof of the type preservation lemma from \cref{subsec-proofDevel}.

\subsection{Discussion}
Both autoformalizations successfully produced results that faithfully
capture a significant share of the underlying document,
with minor problems identified in human reviews.
Even though the AI-authored proof was not optimized for autoformalization and its mechanization required more human feedback identifying alignment errors,
the 
results are of comparable size and incur similar cost in terms of compute resources:
The human-authored proof mechanization spans 1938 lines at an LLM usage cost of \cost{139},
vs.\ 2071 lines and \cost{140} for the AI-authored proof (using \claude 4.6 at a cost of \cost{5} per 1M input and \cost{25} per 1M output tokens). %
%
Resource usage in terms of human work is more difficult to quantify,
but the AI paper proof mechanization demanded more human feedback in the review process,
as well as more rounds of review.
As the proofs themselves are checked for correctness by the proof assistant,
the reviews were focused on aligning theorem statements and definitions with the underlying paper,
where the more formal style of the human expert development proved advantageous.


\paragraph{Claude as a proof engineer.}
The proof documents are mostly idiomatic \isabelle developments:
They use locales, inductive predicates, and recursive data types,
as well as nested \isar proofs~\cite{Wenzel99} of significant length and complexity.
The proofs are generally of high quality and maintainable, yet at times overly verbose,
\eg several induction proofs with explicit case distinctions can easily be compressed into one-line proofs.
Moreover, the agent showed a preference towards object-level quantifiers over structured \isar statements,
making theorems harder to reuse (see \cref{fig:gen_statement} for an example).
The resulting mechanizations 
may not meet the highest quality standards of expert proof engineers, but can serve as a good baseline for further improvements.

The agent routinely ignored the prescribed gradual workflow,
instead producing large \isar proofs at once,
then fixing the resulting errors.
Still, human intervention into the proof engineering process was only required to remedy a misconception on non-terminating proof methods,
where the agent misinterpreted \iq's responses.
Learning from our experience, we provide an improved
\isabelle interaction guide for AI agents in the supplementary material.



\paragraph{Future directions.}
Our experiments are a playground for autoformalization with minimal dependencies and a simple background theory.
In future work, we plan to fully stress the AI agent's capabilities,
exploring navigability in large proof libraries, like the AFP~\cite{afp_website}
and complex proof developments in computer science.
A key challenge is reducing resource cost,
both in terms of compute and human feedback.
This might involve improving utilization of symbolic tools, such as \sledgehammer~\cite{PaulsonB10},
the \isabelle linter~\cite{linter},
and the development of specialized proof engineering LLMs.

\section{
	Autogeneralization based on Human Hints}
\label{sec:independence_system}

Our final experiment explores the AI agent's ability to separate the generic from the ad hoc component of the algorithm (\cref{sec-annotate}).
To this end,
we handed a hint to the AI agent to consider the following
independence system (IS),
which abstracts the type annotation problem:
$(S, F) \text{ where }  F = \{ F' \mid F' \subseteq S \land S \setminus F' \text{ covers } V\}$ and $S \setminus F$ corresponds to the to-be-annotated positions and $V$ to the input term's tyvars.
We then prompted the agent to reprove the results from the AI-authored paper proof (\Cref{sec:llm_paper_proof})
by taking advantage of the generic setting for the standard best-in-greedy algorithm for ISs~\cite[Chapter 13]{KorteVygenOptimisation}.
The AI agent succeeded in reproving the local minimality result
---that was the part benefitting most from the general results on the best-in-greedy algorithm.
The proof was 6.25 pages, computed at a price of \cost{11}
(while having full access to the original proof). 
Next, we prompted the agent to formalize this pen-and-paper proof in Isabelle/HOL, building on an existing Isabelle formalization of 
ISs~\cite{MatroidIsabelle}.
The agent was able to autoformalize the arguments of its own pen-and-paper proofs but, instead of importing the existing formalization, which we asked it to use, proceeded by copying definitions and lemmas.
The formalization (of both local minimality and completeness)
spans 1800 lines, in addition to 200 lines reused from the pre-existing formalization of ISs.
The cost of the entire experiment was \cost{93}.

This experiment was performed by an author with no prior familiarity with the correct printing problem who wanted to understand the problem abstractly by
formulating it in terms of ISs.
The AI agent demonstrated that it can drastically accelerate the process of understanding the problem by pinning down all the details of the connection to ISs
in less than three hours. 
%
Thus, our experiences, as well as those of others,
hint to
a role of LLMs as ``propped up'' automated proof methods or counter-example finders,
benefiting the economics of formal theorem proving and even the practice of mathematics.
We also hypothesize that user-provided hints make (formal) proving using LLMs cheaper.

\section{Related Work and Conclusion}
\label{sec-related}
Type annotations are typically considered as part of \emph{type inference}:
the process of translating a partially typed term from an external language
to a fully typed term in an internal language~\cite{local-type-inf}.
In this work, we consider the reverse direction:
the translation of a fully typed term to a partially typed one.
Since type inference for expressive languages, such as System F, is undecidable~\cite{system-f-infer-undecidable},
type inference algorithms for these languages are only partially complete:
for any well-typed $t$ of type $\tau$,
there is $t'$ with $t\compl t'$ such that running type inference on $t'$ yields $\tau$.
Completeness in this context is also called \emph{annotatability}~\cite{bi-typing-survey}.
Note that annotatability does not tell us \emph{where} type annotations are needed.
Xue and Oliveira~\cite{contextual-type-inf} hence call this \emph{weak annotatability},
differentiating it from \emph{strong annotatability} which also describes where annotations are needed.
Existing work on strong annotatability~\cite{local-contextual-type-inf,contextual-type-inf,bi-check-compl-high-rank,bi-typing-formal,spine-local-type-inf,mlf}
covers the completeness part of the 
correct printing problem
but, to our knowledge,
does not address minimality.
Unlike in our case, the languages used in many of these works do not admit most general well-typed completions.
As a result, annotation criteria are typically tied to the respective type inference algorithm with little control to specify favorable annotation positions,
whereas the approach studied here is independent of the inference algorithm and controllable via the $\pickPos$ function.

A related concern is the construction of correct pretty printers \cite{parsing-printing-unified,printing-correct}.
The focus there is to create infrastructure for correct-by-construction parser and printer combinators,
whereas we are concerned with the correctness of a particular pretty printer
that minimizes type annotations.

For pen-and-paper mathematics, LLMs have shown significant and accelerating progress (see Ju and Dong's survey~\cite{ai4mathsurvey}).
Applications include problems ranging from finding conjectures in research mathematics~\cite{daviesAdvancingMathematicsGuiding2021}, to solving mathematics olympiad problems~\cite{deepMindIMO}, all the way to proving open problems in theoretical computer science~\cite{knuth2026claudescycles}.
Our work here demonstrates the AI's abilities in a new domain of pen-and-paper mathematics: we show that off-the-shelf LLMs can
produce new pen-and-paper theorem statements and proofs in programming language metatheory while based only on informal hints as input, e.g.\ an algorithm lacking formal correctness specifications (as in \cref{sec:llm_paper_proof}) or a pointer to a
pen-and-paper description of a general concept (as in \cref{sec:independence_system}).
Our use case showed that this can be done within reasonable time and financial cost, advocating for LLMs to be used in an interactive loop with human researchers, \ie as a day-to-day proof automation tool, substantially accelerating the human's progress.

The area of autoformalization is rapidly growing.
Several fine-tuned models~\cite{deepseekprover,stepfun,goedelprover}
and sophisticated frameworks~\cite{lego,draft-sketch-prove}
have been developed.
We demonstrated that even a simple setup---using
an off-the-shelf AI coding environment, LLM, and Isabelle MCP server---%
can be sufficient for non-trivial autoformalization tasks.
Our setup was inspired by Urban's remarkable experiment~\cite{urban-topology},
who used an off-the-shelf LLM to autoformalize textbook results in topology
without human intervention.
Notable experiments in the area of programming languages
are by Xi and Odersky~\cite{agentic-proof-case-study}
and Ilya Sergey~\cite{move-borrow}.
Both present an iterative development
where humans wrote the definitions while agents handled the proof development.
This is unlike our work, which hands all steps to the LLM,
only providing human review once a complete formalization is obtained.

\paragraph*{Conclusion}\label{sec:conclusion}
\looseness=-1
We provide a formal account of the problem of sufficient and minimal type annotations for correct printing, 
in particular the Smolka-Blanchette algorithm,
with mechanized statements and correctness proofs. 
A main focus was on using LLMs as research assistants in the context of programming language research. 
Our case study reveals that state-of-the-art models and human insights can be complementary and therefore mutually beneficial.

Our work suggests that LLMs can already be used for day-to-day proof assistance, in the same way as existing Isabelle tools like Sledgehammer, Nitpick, and auto, but with a substantially broader scope and power. 
We also point to a promising future in which a deep integration of AI agents and proof assistants could catalyze both AI's and humans' capabilities for proving theorems.

\begin{credits}
\subsubsection*{\ackname} This research is conducted in the Copilots for Isabelle project under the AI for Math Fund, an initiative by Renaissance Philanthropy with support from XTX Markets. Our interest in adding type annotations during printing arose while implementing a copilot to translate 
apply-style proofs to readable Isar proofs. 
%

\subsubsection*{\discintname}
The authors have no competing interests to declare.
\end{credits}

\bibliographystyle{splncs04}
\bibliography{bib,long_paper}

\ifwithappendix
\appendix
\gdef\theHsection{\Alph{section}}

\begin{center} 
{\Large APPENDIX}
\end{center}

This appendix contains further details on the concepts and statements from the main paper. 
In case of acceptance, the appendix will be replaced by a technical report cited from the paper 
and made available online. 

\section{More Details on Syntax}
\label{app-detailsSyntax} 

The definition of substitution application $\_[\rho]$ on types consists of the following recursive equations:
\begin{itemize}
	\item $\alpha[\rho] = \rho(\alpha)$; 
	\item $(\sigma \Ra \tau)[\rho] = \sigma[\rho] \Ra \tau[\rho]$; 
	\item $((\sigma_1,\ldots,\sigma_n)\,\kappa)[\rho] = (\sigma_1[\rho],\ldots,\sigma_n[\rho])\,\kappa$. 
\end{itemize}

And similarly for the one on terms: 
\begin{itemize}
	\item $x_\xi[\rho] = x_{\xi[\rho]}$
	\item $c_\xi[\rho] = c_{\xi[\rho]}$
	\item $(t_1\,t_2)_\xi[\rho] = (t_1[\rho]\,t_2[\rho])_{\xi[\rho]}$ 
	\item $(\lambda x_\xi.\,t)_\zeta[\rho] = (\lambda x_{\xi[\rho]}.\,t[\rho])_{\zeta[\rho]}$ 
\end{itemize}

The definition of $\TV$ (the occurring tyvars operator) on types consists of the following recursive equations: 
\begin{itemize}
	\item $\TV(\alpha) = \{\alpha\}$; 
	\item $\TV(\sigma \Ra \tau) = \TV(\sigma) \cup \TV(\tau)$; 
	\item $\TV((\sigma_1,\ldots,\sigma_n)\,\kappa) = \bigcup_{i=1}^n \TV(\sigma_i)$. 
\end{itemize}

And similarly for the one on terms: 
\begin{itemize}
	\item $\TV(x_\xi) = \TV(\xi)$
	\item $\TV(c_\xi) = \TV(\xi)$
	\item $\TV((t_1\,t_2)_\xi) = \TV(t_1) \cup \TV(t_2) \cup \TV(\xi)$ 
	\item $\TV((\lambda x_\xi.\,t)_\zeta) = \TV(\xi) \cup \TV(t) \cup \TV(\zeta)$ 
\end{itemize}

\bigskip
We write $[n_1,\ldots,n_k]$
for the list consisting of $n_1,\ldots,n_k$ (in particular, $[]$ for the empty list)
and $n \cdot l$ for consing an element $n$ to a list $l$.
The set of positions of a term, $\Poss(t)$, is defined by:
\begin{align*}
	&\Poss(x_\xi) = \Poss(c_\xi) = \{[]\}
	\qquad \Poss((t\,s)_\xi) = \{[], 1 \cdot \Poss(t), 2 \cdot \Poss(s)\}\\
	&\Poss((\lambda x_\xi.\;t)_\zeta) = \{[], [1], 2 \cdot \Poss(t)\}
\end{align*}

Each position of a term uniquely identifies a location of a subterm or a binding variable. If $p \in \Poss(t)$, the \emph{maybe-type of (the subterm at) position $p$ in $t$}, $\mtpOf(t,p)$, is defined recursively:
\vspace*{-2ex}
\begin{align*}
	&\mtpOf(x_\xi,p) = \xi \qquad &\mtpOf((t_1\,t_2)_\xi,p) =
	\begin{cases}
		\mtpOf(t_1,p'), & \text{if } p = 1 \cdot p' \\
		\mtpOf(t_2,p'), & \text{if } p = 2 \cdot p' \\
		\xi, & \text{otherwise}
	\end{cases}\\
	&\mtpOf(c_\xi,p) = \xi \qquad&\mtpOf((\lambda x_\xi.\,t)_\zeta,p) =
	\begin{cases}
		\xi,   & \text{if } p = [1] \\
		\mtpOf(t,p'), & \text{if } p = 2 \cdot p' \\
		\zeta, & \text{otherwise}
	\end{cases}
\end{align*}

The annotation deletion operator $t[p:=\bot]$ has the following recursive definition:
\begin{itemize}
	\item[] $x_\xi[p:= \bot]$ $=$ $x_\bot$
	\hspace*{10ex}  $c_\xi[p:= \bot]$ $=$ $c_\bot$
	\item[] $(t_1\,t_2)_\xi[p:= \bot]$ $=$ 
	$\left\{\begin{array}{ll}
		(t_1[p':= \bot]\,t_2)_\xi, & \mbox{ if $p$ has the form $1 \cdot p'$}
		\\
	 (t_1\,(t_2[p':= \bot]))_\xi, & \mbox{ if $p$ has the form $2 \cdot p'$}
		\\
		(t_1\,t_2)_\bot, \mbox{ otherwise}
	\end{array} \right.$  
	\item[] $(\lambda x_\xi.\,t)_\zeta[p:= \bot]$ $=$  
	$\left\{\begin{array}{ll}
		(\lambda x_\bot.\;t)_\zeta, & \mbox{ if $p=1$}
		\\
		(\lambda x_\xi.\;(t[p':= \bot]))_\zeta, & \mbox{ if $p$ has the form $2 \cdot p'$}
		\\
		(\lambda x_\xi.\,t)_\bot, \mbox{ otherwise}
	\end{array} \right.$   
\end{itemize}

\section{More Details on the Smolka-Blanchette Algorithm}
\label{app-moreSmplkaBlanchette} 

The definition of $\decrease$ is correct (i.e., its recursion terminates) because, in the recursive call,
the number of positions $p$ in $s$ such that $\mtpOf(s,p) \not = \bot$, \ie
$|\Poss_{\not=\bot}(s)|$ where $\Poss_{\not=\bot}(s) = \{p \in \Poss(s) \mid\mtpOf(s,p) \not = \bot\}$, decreases by 1. %
\looseness=-1

\section{More Details on the Human-Authored Paper Proof Development}
\label{app-detailsHumanProofs}

\subsection{Constructor-aware inversion lemmas} 
\label{appsubsec-inversions}

\begin{mylemma}\rm \label{lem-invTyping} \rm (term-constructor-guided inversion rules for $\vdash$)
	The following hold: 
	\\(1) If $\vdash c_\sigma$ then $\sigma \leq \ctpOf(c)$  
	\\(2) If $\vdash (t_1\,t_2)_\sigma$ then there exists $\tau$ such that $\vdash t_1$, $\vdash t_2$ and 
	$\tpOf(t_1) = \tpOf(t_2) \Ra\sigma$. 
	\\(3) If $\vdash (\lambda x_\sigma.\,t)_\tau$ then $\vdash t$ and $\tau = \sigma \Ra \tpOf(t)$. 
\end{mylemma} 
\begin{proof} 
	Follows from the standard inversion rule associated to the inductive definition 
	of $\vdash$ and the injectiveness and distinctness of the syntactic term constructors. 
	\qed 
\end{proof}

\begin{mylemma}\rm \label{lem-leftInvCompl} \rm (term-constructor-guided left inversion rules for $\compl$)
	The following hold: 
	\\(1) If $x_\xi \compl t$ then there exists $\xi'$ such that $\xi \compl \xi'$ and $t = x_{\xi'}$. 
	\\(2) If $c_\xi \compl t$ then there exists $\xi'$ such that $\xi \compl \xi'$ and $t= c_{\xi'}$. 
	\\(3) If $(s_1\,s_2)_\xi \compl t$ then there exist $s_1'$, $s_2'$ and $\xi'$ such that $s_1 \compl s_1'$, $s_2 \compl s_2'$, 
	$\xi \compl \xi'$ and $t= (s_1'\,s_2')_{\xi'}$. 
	\\(4) If $(\lambda x_\zeta.\;s)_\xi \compl t$ then there exist $s'$, $\zeta'$ and $\xi'$ 
	such that $s \compl s'$,  
	$\zeta \compl \zeta'$, $\xi\compl \xi'$ and $t= (\lambda {x_{\zeta'}}.\;s')_{\xi'}$. 
\end{mylemma} 
\begin{proof} 
	Follows from the standard inversion rule associated to the inductive definition 
	of $\compl$ and the injectiveness and distinctness of the syntactic term constructors. 
	\qed 
\end{proof} 

\begin{mylemma}\rm \label{lem-rightInvCompl} \rm (term-constructor-guided right inversion rules for $\compl$)
	The following hold: 
	\\(1) If $t \compl x_\xi$ then there exists $\xi'$ such that $\xi' \compl \xi$ and $t= x_{\xi'}$. 
	\\(2) If $t \compl c_\xi$ then there exists $\xi'$ such that $\xi' \compl \xi$ and $t= c_{\xi'}$. 
	\\(3) If $t \compl (s_1\,s_2)_\xi$ then there exist $s_1'$, $s_2'$ and $\xi'$ such that $s_1' \compl s_1$, $s_2' \compl s_2$, 
	$\xi' \compl \xi$ and $t= (s_1'\,s_2')_{\xi'}$. 
	\\(4) If $t \compl (\lambda x_\zeta.\;s)_\xi$ then there exist $s'$, $\zeta'$ and $\xi'$ 
	such that $s' \compl s$,  
	$\zeta' \compl \zeta$, $\xi'\compl \xi$ and $t= (\lambda {x_{\zeta'}}.\;s')_{\xi'}$. 
\end{mylemma} 
\begin{proof} 
	Follows from the standard inversion rule associated to the inductive definition 
	of $\compl$ and the injectiveness and distinctness of the syntactic term constructors. 
	\qed 
\end{proof}

\subsection{Lemmas on syntax basics} 
\label{appsubsec-syntaxBasics}

\begin{mylemma}\rm \label{lem-typeBasics}
	The following hold:  
	\\(1) $\tau[\rho \scmp \rho'] = \tau[\rho'][\rho]$. 
	\\(2) If $\beta\notin \TV(\tau)$, then $\tau[\beta/\alpha][\alpha/\beta] = \tau$. 
	\\(3) $\tau[\rho] = \tau[\rho']$ iff $\forall \alpha \in \TV(\tau).\,\rho(\alpha) = \rho'(\alpha)$. 
	\\(4) $\TV(\tau[\rho]) =  \bigcup_{\alpha \in \TV(\tau)} \!\TV(\rho(\alpha))$.
	\\(5) $\TV(\tau)$ is finite. 
	\\(6) $\leq$ is a preorder on types. 
	\\(7) Assume that $\forall \alpha \in \TV(\tau).\;\rho(\alpha) \leq \rho'(\alpha)$. 
	Then $\tau[\rho] \leq \tau[\rho']$.  
\end{mylemma} 
\begin{proof}
	(1)--(5) and (7) follow by straightforward structural induction over $\tau$, using the definitions of $\TV$ and $\_[\_]$ on types. (3) also needs the injectiveness and distinctness properties of the syntactic constructors for types. 
	\\
	(6) follows immediately from (1). 
	\qed
\end{proof}

\begin{mylemma}\rm \label{lem-termBasics}
	The following hold for all unambiguous terms $t$:  
	\\(1) $t[1_\TVar] = t$ and $t[\rho \scmp \rho'] = t[\rho'][\rho]$. 
	\\(2) If $\beta\notin \TV(t)$, then $t[\beta/\alpha][\alpha/\beta] = t$. 
	\\(3) $t[\rho] = t[\rho']$ iff $\forall \alpha \in \TV(t).\,\rho(\alpha) = \rho'(\alpha)$.\footnote{This ``iff'' version covers both substitution extensionality and its converse.} 
	In particular, $t[\rho] = t$ iff $\forall \alpha \in \TV(t).\,\rho(\alpha) = \alpha$. 
	\\(4) $\TV(t[\rho]) = \bigcup_{\alpha \in \TV(t)} \!\TV(\rho(\alpha))$.
	\\(5) $\TV(t)$ is finite. 
	\\(6) $\leq$ is a preorder. 
	\leftOut{These are not needed for the unoptimized verion: 
		, $\equiv$ an equivalence, and $<$ a strict partial order on terms. 
	\\(7) Assume that $\forall \alpha \in \TV(t).\;\rho(\alpha) \leq \rho'(\alpha)$. 
	Then $t[\rho] \leq t[\rho']$. 
	\\(8) Assume that $\forall \alpha \in \TV(t).\;\rho(\alpha) \leq \rho'(\alpha)$ and $\exists \alpha \in \TV(t).\;\rho(\alpha) < \rho'(\alpha)$. 
	Then $t[\rho] < t[\rho']$. 
	\\ 
	(9)  $t[\rho]$ is unambiguous iff 
	$t$ is unambiguous and 
	$\rho(\alpha)$ is unambiguous for all $\alpha\in \TV(t)$.  
} 
\end{mylemma} 
\begin{proof}
	(1)--(5) 
	follow by straightforward structural induction over $t$, using the definitions of $\TV$ and $\_[\_]$ on terms,
	and the corresponding fact for types (from Lemma~\ref{lem-typeBasics}).  Again, (3) needs injectiveness and distinctness of the syntactic constructors for terms.  The second part of (3) follows from the first part of (3) and the first part of point (1). 
	\\
	(6) again follows easily from (1). 
	\qed
\end{proof}

\subsection{Annotation lemmas} 
\label{appsubsec-annotationLemmas}

\begin{mylemma}\rm \label{lem-basicsPosUpdate}
	Assume $t\in \Term$ and $p\in\Poss(t)$. Then the following hold:
	\\(1) $t[p:= \bot] \compl t$.
	\\(2) $\mtpOf(t[p:=\bot],p) = \bot$. 
	\\(3) For all $q \in \Poss(t) \sm \{p\}$, $\mtpOf(t[p:=\bot],q) = \mtpOf(t,q)$.
	\\(4) If $t$ is unambiguous, then $t[p:= \bot]$ is unambiguous. 
\end{mylemma} 
\begin{proof}
	By straightforward induction on $t$, using the definitions of $\_[\_:=\_]$ and $\mtpOf$, and that of unambiguity. 
	\qed
\end{proof}

\subsection{The type preservation lemma} 
\label{appsubsec-typepresLemma}

\begin{mylemma} \rm \label{lem-subst}  
	Assume $v$ is an unambiguous F-term such that $\vdash v$ and $\rho$ a substitution.
	Then $\vdash v[\rho]$. 
\end{mylemma} 
\begin{proof}
	By structural induction on $v$, using the definitions of $\TV$ and $\_[\_]$. \qed 
\end{proof}

\subsection{The$\compl$- and $(\compl,\leq)$-lemmas} 
\label{appsubsec-complLemmas}

\begin{mylemma}\rm \label{lem-complTrivial}
	The following hold: 
	\\
	(1) $\compl$ is a partial order (and therefore $\complS$ is a strict partial order) on $\Term$ 
	having the U-terms as the minimal elements and the F-terms as the maximal elements. 
	\\   
	(2) If $t \compl s$, then $\mtpOf(t) \compl \mtpOf(s)$. 
	\\ 
	(3) $\er(t) \compl t$.  
	\\
	(4) If $t \compl s$ then $t$ is unambiguous iff $s$ is unambiguous. 
\end{mylemma}
\begin{proof}
	(1), (2) and (4): By induction on the definition of $\compl$. 
	For transitivity, we also need the right-inversion rules for $\compl$ (Lemma~\ref{lem-rightInvCompl}).
	
	\noindent 
	(3): By structural induction on $t$, using the definition of $\er$.  
\end{proof}

The next lemma essentially says that any two terms connected by $\compl$ or $\leq$ have the ``same shape'' hence the same positions, and their ``at position'' type annotations are related correspondingly. 
\begin{mylemma} \rm \label{lem-sameShape} 
	(1) If $t \compl s$ then $\Poss(t) = \Poss(s)$.  
	\\
	(2) If $t \compl s$ and $p \in \Poss(t)$ then  $\mtpOf(t,p) \compl \mtpOf(s,p)$.
	\\
	(3) If $t \leq s$ then $\Poss(t) = \Poss(s)$. 
	\\
	(4) If $t \leq s$ and $p \in \Poss(t)$ then  $\mtpOf(t,p) \leq \mtpOf(s,p)$.
\end{mylemma}
\begin{proof}
	(1) and (2): By induction on the inductive definition of $\compl$, using the definitions of $\Poss$ and $\mtpOf(t,p)$. 
	
	(3) and (4): We obtain $\rho$ such that $t = s[\rho]$. Then the proof goes by structural induction on $t$. 
	\qed 
\end{proof}

\subsection{The generic reverse greedy lemmas} 
\label{appsubsec-genericReverseGreedy}

Strictly speaking, only Lemma~\ref{lem-noTest} is entirely generic, 
whereas Lemma~\ref{lem-decreaseCompl} also uses some specific properties of $\compl$, 
including transitivity. 

\begin{mylemma}\rm \label{lem-noTest}
	Assume $t$ is an unambiguous term and $v$ is an F-term 
	and $\pickPos \in \Compat$. 
	Let $s = \decrease(\pickPos,v,t)$. 
	Then $\neg\,\exists p.\;\coverageTest(v,s,p)$. 
\end{mylemma} 
\begin{proof}
	Immediately by induction on $|\Poss_{\not=\bot}(s)|$ (alternatively, by the computation induction principle stemming from 
	the definition of $\decrease$).  \qed
\end{proof}

\begin{mylemma}\rm \label{lem-decreaseCompl}
	Assume $t$ is an unambiguous term and $v$ is an F-term 
	and $\pickPos \in \Compat$. 
	Let $s = \decrease(\pickPos,v,t)$. 
	Then $s \compl \mgen(t)$, in particular $s$ 
	is unambiguous. 
\end{mylemma} 
\begin{proof}
	By induction on $|\Poss_{\not=\bot}(s)|$ (alternatively, by the computation induction principle stemming from 
	the definition of $\decrease$), 
	using Lemma~\ref{lem-complTrivial}(1,4).  \qed
\end{proof}

\subsection{Some omitted proofs and proof sketches} 
\label{appsubsec-omittedProofs}

\noindent
\textbf{Proof of Lemma~\ref{lem-otherTriv}.} 
By induction on the definition of $\compl$, using the definition of $\mtpOf$. \qed

\medskip
\noindent
\textbf{Proof of Lemma~\ref{lem-multipleCompl}.} 
By structural induction on $v$, using the definitions of $\compl$ and $\_[\_]$. \qed

\medskip
\noindent
\textbf{Proof of Lemma~\ref{lem-twoMG}.} 
	From the definitions, we have $s \compl v$ and $\vdash v$, and we obtain $\alpha$ such that $\alpha \in \TV(v) \sm \TV(s)$. 
	From this and $s \compl v$, using  the annotation lemmas and the  $\compl$-lemmas, we have
	(a) $\forall p \in \Poss(v) = \Poss(s).\;\alpha \in \TV(v,p) \implies \mtpOf(s,p) = \bot$.
	
	We choose $\beta$ such that (b) $\beta \notin \TV(v)$ (in particular, $\beta \not= \alpha$), and let $v' = v[\beta/\alpha]$. By  syntax basics and $\alpha \in \TV(v)$, we have $v\not=v'$. Moreover, again by syntax basics, from (b) we obtain (c) $v = v[\beta/\alpha][\alpha/\beta] = v'[\alpha/\beta]$.
	
	By the type preservation lemma from  $\vdash v$
	we obtain $\vdash v'$. Moreover,
	by Lemma~\ref{lem-multipleCompl}
	used with the identity and $\beta/\alpha$ as substitutions, from
	$s \compl v$ and (a) we obtain $s \compl v'$.
	Thus, $v'$ is distinct from $v$, and is a well-typed completion of $s$.
	\looseness=-1
	
	It remains to show that $v'$ is a most general well-typed completions of $s$.
	Let $w$ be a well-typed completion of $s$. Since $v$ is a most general completion, we obtain $\rho$ such that $w = v[\rho]$. Moreover, by syntax basics, from (c) we obtain
	$v[\rho]= v'[\alpha/\beta][\rho] = v'[\rho \scmp (\alpha/\beta)]$. So $w \leq v'$, as desired.
	\qed

\section{More Details on the AI Paper Development}
\label{app-detailsAganticPaper}

The human and AI agent pen-and-paper developments have been developed mostly independently, 
with \emph{the only influence occurring from the AI agent to the human}---who, as we report in the main paper, obtained the idea of stating and proving minimality from the AI agent.  Unavoidably, independent developments lead to different design decisions, and this was the case here: 
\begin{itemize}
	\item[(1)]  The human expert considered the problem starting from an arbitrary typable C-term $t$, (which corresponds to an Isabelle term, i.e., features the usual Church-style type annotations), then proved in Prop~\ref{prop-Cterm-uniqueCompleteness} that such a term has a unique well-typed completion, namely the F-term $\mgen(t)$,  and applied the algorithm starting with $\mgen(t)$. By contrast, the AI ignored the original term $t$ and worked directly with an F-term (corresponding to what the human called $\mgen(t)$).\footnote{Incidentally, since the AI did not prove a result analogous to Prop~\ref{prop-Cterm-uniqueCompleteness}, it did not need syntactic-constructor-aware inversion rules or anything equivalent to that.} In this respect, the result proved by the AI is less complete, but consistent 
	because $\er(\mgen(t))$ is the same as $\er(t)$.
	\item[(2)] The human expert did not consider typing contexts, after making an argument that these are ``not interesting'' in that the problem can be reduced to an empty-context problem (by considering the typing context to be part of the signature). By contrast, the AI considered typing contexts, and formulated the results with the additional requirements for the proper treatment of tyvars from the  context by the substitutions. In this respect, the result proved by the AI is more complete, as it avoids any informal (meta-)argument in order to cover the general case.
	 \item[(3)] In the human account, the algorithm operates with an evolving (partially annotated) term of the same shape as the original term, which iteratively loses annotations. In the AI account, the algorithm, and the corresponding theorems, operate instead on sets of positions, and therefore use inclusions between sets of positions as opposed to the annotation subsumption relation $\compl$. This makes the AI account more direct, although arguably less intuitive.\footnote{Incidentally, the AI, following the implementation, represents positions differently from (but equivalently to) the human, namely as numbers rather than lists.}
	  \item[(4)] The AI fixed a lot of items and made several assumptions globally, which were shared by many lemmas and proposition statements, including the main results (in the style of Isabelle locales), whereas the human preferred to distribute the assumptions as needed in each statement. These two approaches trade succinctness on the one hand with readability and precision on the other hand.%
	  \footnote{Incidentally, the AI paper introduced some scoping issues, such as a definition inlined inside a theorem and referred to from another theorem, which were nevertheless sorted at autoformalization time. It also introduced some useless/tautological but harmless assumptions.} 
\end{itemize}

\section{More Details on the Autoformalization Experiments}
\label{app-detailsAutoformalization}

We show the complete statements of the formalization in \Cref{app-isabelleAI} and \Cref{app-isabelleHuman}, for the AI-authored and human-authored proofs respectively.

\subsection{AI-authored proof}
\label{app-isabelleAI}

\newcommand{\isafigscale}{0.35}

The main results of the formalization are stated within the context of the locale \texttt{annotation\_problem}~(cf. \cref{fig:ai_locale}).
It introduces
\begin{itemize}
\item the fully annotated input term \texttt{a}, corresponding to $t$ in the definition of $\smobla$; 
\item a most general completion \texttt{a\_star} (of \texttt{a} stripped from type annotations), corresponding to $\mgen(\er(t))$;
\item a type substitution relating \texttt{a} and \texttt{a\_star};
\item a context $\mathtt{\Gamma}$ and a constant signature \texttt{const\_type};
\item well-typedness assumptions on \texttt{a} and \texttt{a\_star};
\item a freshness assumption on the type substitution, modeling that type inference uses fresh type variables.
\end{itemize}
Both completeness (cf. \cref{fig:ai_completeness}) and minimality (cf. \cref{fig:ai_minimality}) are stated within the \texttt{annotation\_problem} locale. 
A limitation inherited from the AI-authored paper is that both statements refer to the set of annotation positions, not only the output term \texttt{t\_out}.

\begin{figure}[t]
\includegraphics[scale=\isafigscale]{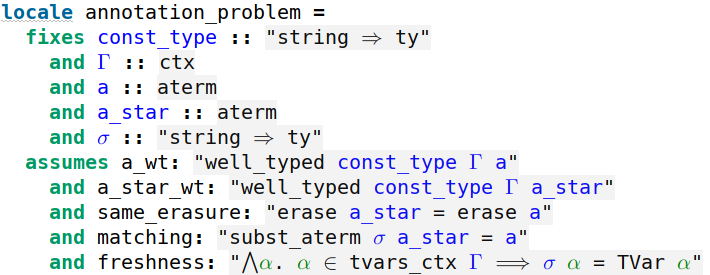}
\caption{AI-based annotation problem locale: well-formedness assumptions on the input \texttt{a} as well as type inference in terms of \texttt{a\_star}.}
\label{fig:ai_locale}
\bigskip
\includegraphics[scale=\isafigscale]{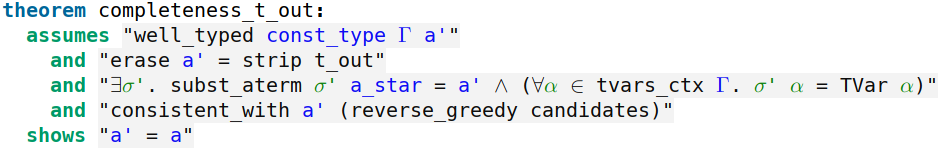}
\caption{\looseness=-1 AI-based completeness statement: 
\texttt{a} is the unique term \texttt{a'} that is \emph{consistent with} \texttt{a} at the positions determined by the \texttt{reverse\_greedy} algorithm.
The locale context only assumes type inference properties for \texttt{a},
so this assumption is repeated here for \texttt{a'}.}
\label{fig:ai_completeness}
\bigskip
\includegraphics[scale=\isafigscale]{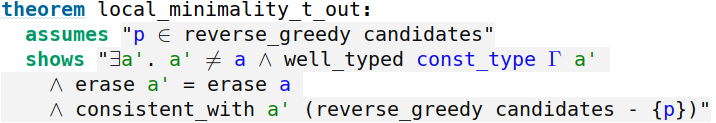}
\caption{AI-based minimality statement:
removing any position the algorithm deems necessary allows us to obtain a well-typed term different from \texttt{a} that agrees with \texttt{a} on the remaining positions.
}
\label{fig:ai_minimality}
\end{figure}

\subsection{Human-authored proof}
\label{app-isabelleHuman}

The formalization defines the locale \texttt{signature\_with\_mgen}~(cf. \cref{fig:human_signature_locale}),
it fixes a type signature and also follows the paper and assumes Thm.~\ref{prop-assumed}.
The \texttt{algorithm} locale (cf. \cref{fig:human_algorithm_locale}) abstracts over specific $\pickPos$ functions,
where the \texttt{compat} assumption corresponds to $\Compat$.
Both completeness (cf. \cref{fig:human_completeness}) and minimality (cf. \cref{fig:human_minimality}) are stated within this \texttt{algorithm} locale. 

\begin{figure}[t]
\includegraphics[scale=\isafigscale]{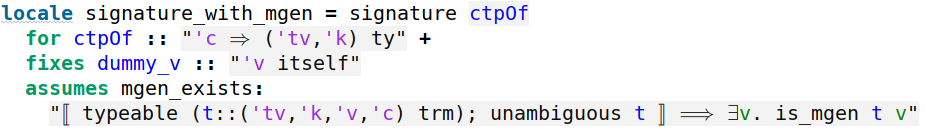}
\caption{Human-based signature locale: corresponds to Thm.~\ref{prop-assumed}}
\label{fig:human_signature_locale}
\bigskip
\includegraphics[scale=\isafigscale]{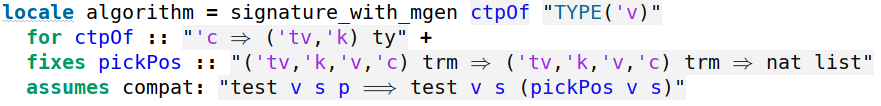}
\caption{Human-based algorithm locale}
\label{fig:human_algorithm_locale}
\bigskip
\includegraphics[scale=\isafigscale]{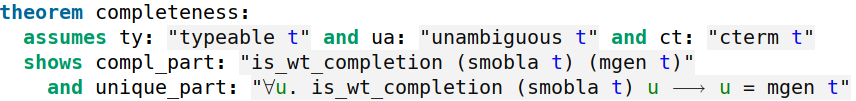}
\caption{Human-based completeness statement: corresponds to Thm.~\ref{thm-correctPrintingHuman}}
\label{fig:human_completeness}
\bigskip
\includegraphics[scale=\isafigscale]{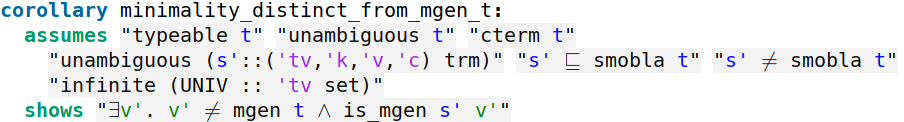}
\caption{Human-based minimality statement: corresponds to Thm.~\ref{thm-correctPrintingMinimalHuman}}
\label{fig:human_minimality}
\end{figure}

\fi

\end{document}